\documentclass[reprint,aps,prb, floatfix,amsmath,amssymb,superscriptaddress]{revtex4-2}

\usepackage{graphicx}
\usepackage{dcolumn}
\usepackage{bm}
\usepackage{xcolor} 
\usepackage{subfigure}
\usepackage{amsthm}
\usepackage{ulem}
\usepackage[colorinlistoftodos]{todonotes}
\usepackage{hyperref}
\usepackage{CJK} 

\DeclareMathOperator{\supp}{supp}
\DeclareMathOperator{\diam}{diam}
\DeclareMathOperator{\dist}{dist}
\DeclareMathOperator{\tbs}{\backslash}
\newcommand{\mtude}[1]{\lvert#1\rvert}
\newtheorem{theorem}{Theorem}
\newtheorem{definition}{Definition}
\newtheorem{lemma}{Lemma}

\newcommand*{\ket}[1]{|#1\rangle}
\newcommand*{\bra}[1]{\langle#1|}
\newcommand*{\norm}[1]{||#1||}
\newcommand*{\bignorm}[1]{\bigg|\bigg|#1\bigg|\bigg|}
\newcommand*{\Fst}[6]{\left[F^{#1#2#3}_#4\right]_{#5#6}}
\newcommand*{\Fbent}[6]{\left[F^{#1#2}_{#3#4}\right]_{#5#6}}
\newcommand*{\Fbtos}[6]{\Fst{#3}{#5}{#2}{#6}{#1}{#4}}

\begin{document}

\begin{CJK*}{UTF8}{min}


\title{Isometric Tensor Network Representation of String-Net Liquids}

\author{Tomohiro Soejima (副島智大)}
\affiliation{Department of Physics, University of California, Berkeley, CA 94720, USA}
\author{Karthik Siva}%
\affiliation{Department of Physics, University of California, Berkeley, CA 94720, USA}
\author{Nick Bultinck}
\affiliation{Department of Physics, University of California, Berkeley, CA 94720, USA}
\author{Shubhayu Chatterjee}
\affiliation{Department of Physics, University of California, Berkeley, CA 94720, USA}
\author{Frank Pollmann}
\affiliation{Technische Universit\"at M\"unchen, Physics Department T42, 85747 Garching, Germany}
\affiliation{Munich Center for Quantum Science and Technology (MCQST), Schellingstr. 4, D-80799 M\"unchen}
\author{Michael P. Zaletel}
\affiliation{Department of Physics, University of California, Berkeley, CA 94720, USA}
\date{\today}

\begin{abstract}
Recently, a class of tensor networks called isometric tensor network states (isoTNS) was proposed which generalizes the canonical form of matrix product states to tensor networks in higher dimensions. While this ansatz allows for efficient numerical computations, it remained unclear which phases admit an isoTNS representation. In this work, we show that two-dimensional string-net liquids, which represent a wide variety of topological phases including discrete gauge theories, admit an exact isoTNS representation. We further show that the isometric form can be preserved after applying a finite depth local quantum circuit. Taken together, these results show that long-range entanglement by itself is not an obstruction to isoTNS representation and suggest that all two-dimensional gapped phases with gappable edges admit an isoTNS representation. 
\end{abstract}

\maketitle

\end{CJK*}



\section{Introduction}
A central challenge in quantum many-body physics is to accurately and efficiently represent wavefunctions and evaluate expectation values of observables and correlation functions. The matrix product state (MPS) ansatz has proven to be a computationally efficient and accurate way of representing ground states of many Hamiltonians in 1D. The remarkable success of this ansatz can be attributed to three key features. First, ground states of gapped Hamiltonians in 1D exhibit a so-called ``area law" entanglement entropy, a property which is built into the variational class of MPS with finite bond dimension \cite{Hastings07,Arad12}. Second, by exploiting a gauge freedom inherent to tensor network representations of many-body states, MPS can be put in a canonical form that allows for an efficient calculation of expectation values of local operators. Third, there exist algorithms such as the Density Matrix Renormalization Group (DMRG) \cite{White92,Schollwock11}, Time-Evolving Block Decimation (TEBD) \cite{Vidal03,Vidal04} and the Time-Dependent Variational Principle (TDVP) \cite{Haegeman11}, which can efficiently explore the MPS manifold and optimize the variational parameters contained in the local tensors that generate the ground state ansatz \cite{Arad17}. 

Tensor networks for systems in higher dimensions have been studied, such as Projected Entangled Pair States (PEPS) \cite{Nishino96, Verstraete04}. However, without a canonical form, tensor contraction of a general network is \#P-hard \cite{Schuch07}. Consequently, evaluating reduced density matrices is either costly or can only be done approximately using e.g. the corner transfer matrix \cite{Baxter68,Nishino96,Orus09,Corboz16,Vanderstraeten16} and boundary MPS methods \cite{Verstraete04,Jordan08,Vanderstraeten16}, or the tensor renormalization approach \cite{Levin07,Xie09,Evenbly15,Yang15}. 

Recently, subclasses of 2D tensor networks, motivated by the canonical form for MPS, have been proposed \cite{Zaletel19, Haghshenas19}.  The ``isometric tensor network state'' (isoTNS) defined in Ref.~\cite{Zaletel19} generalizes the isometry conditions of the canonical form of MPS to higher dimensional tensor networks and generalizes the orthogonality center to an orthogonality hypersurface (see App.~\ref{app:isoTNS} for details). As a result, computations in this ansatz are significantly faster than in unconstrained tensor networks. The cost of the full-update, for example, is reduced from $O(\chi^{12}) \to O(\chi^6)$ \cite{Orus09}, where $\chi$ is the dimension of the tensors. However, as isoTNS are a restricted subclass of PEPS, it remains unclear what quantum phases they can accurately represent. It is crucial to understand this point because if there is a property common among ground states of interest which forbids their representation as an isoTNS, then tensor network calculations will remain biased even as the bond dimension increases.

In this paper, we pursue this line of investigation by characterizing the representative power of the isoTNS ansatz. We first find numerically that a family of states in the toric code phase can be represented by an isoTNS. Motivated by this result, we then prove analytically that the string-net liquid states introduced in Ref.~\cite{Levin05} admit \textit{exact} isoTNS representations with finite bond dimension, and further show that states generated from these fixed-point wavefunctions by finite-depth quantum circuits also admit exact isoTNS descriptions.
More specifically we prove that:
\begin{enumerate}
    \item For every string-net liquid model as defined in \cite{Levin05} and generalized in \cite{Lin14}, there exists an exact representation of the fixed-point ground state wavefunction as an isoTNS with finite bond dimension, such that the orthogonality hypersurface may be placed anywhere.
    \item All wavefunctions which may be transformed to string-net liquid fixed points by local unitary circuits of finite depth also admit exact isoTNS representations of finite bond dimension.
\end{enumerate}

The states generated by  finite-depth circuits acting on  string-net wavefunctions include a large class of 2D quantum phases, including (but not limited to) all bosonic abelian topological orders with gappable edges \cite{Lin14}. It is widely believed (though to our knowledge not proven) that all bosonic gapped phases with a gappable edge have a string-net representation. Taken together, these results show that long-range entanglement does not form an obstruction for isoTNS representations, and suggest that the ground states of gapped Hamiltonians with gappable edges can be efficiently represented as an isoTNS 
Note, however, that this result does not include phases with chiral topological order, such as the integer quantum Hall effect or Laughlin states, which feature gapless chiral edge states. This is consistent with the folklore that tensor networks with exponentially decaying correlations cannot represent states with chiral topological order \cite{Dubail15,Wahl}. More generally, it does not include models with ungappable edges, which need not be chiral \cite{Levin13}.


The paper is organized as follows. In Sec.~\ref{sec:isoTNS_intro} we briefly review the motivation and definition of the isoTNS ansatz and refer to App.~\ref{app:isoTNS} for more details and review of graphical notation. In Sec.~\ref{sec:numerics} we numerically show that a family of states in the toric code phase, which has topological order, can be efficiently represented using isoTNS ansatz. In Sec.~\ref{sec:string_net_tensors} we state a standard definition of the string-net liquid wavefunction in terms of $F$-symbols and review its tensor network representation. Those who are more interested in the physical implications and not in the technical details can stop at this section and skip ahead to the conclusion.  In Sec.~\ref{sec:string_isometry} we construct a new string-net tensor network which satisfies the local isometry conditions of the isoTNS ansatz. As explained further below, the main problem involves promoting the unitarity of the $F$-symbols in restricted subspaces to  an isometry in the relevant ancilla space of a single tensor. The main technical tools used are a method for placing the fusion constraints along the orthogonality hypersurface and a diagrammatic notation to help show that the the tensors satisfy the isometry criteria and that the resulting wavefunction is unchanged. In Sec.~\ref{sec:finite_depth} we build on this result to show that all states which may be transformed in the string-net liquid fixed points by local unitary circuits of finite depth also admit exact isoTNS representations of finite bond dimension. The main observation used is that by applying a coarse-graining and fine-graining transformation on the tensors, one can restore any isometry condition after application of a unitary operation. Finally, we conclude with speculation about the implications of this result as well as future directions and open questions. The Appendices in this paper collect the notations, constructions, and identities relevant for this work as well as technical details and examples.

\section{\label{sec:isoTNS_intro}Isometric Tensor Network States}
The isoTNS ansatz is a subclass of tensor network states satisfying particular isometry constraints that allow for efficient computation of expectation values of observables. This ansatz lifts the concept of ``canonical forms" used for MPS to tensor networks of higher dimension. Readers unfamiliar with MPS may consult App.~\ref{app:isoTNS} for a more detailed account of the connection between isoTNS and MPS.

We now review the terminology, define the PEPS ansatz, and state the isometry constraints that define the isoTNS ansatz for 2D states.

\begin{figure}
    \centering
    \subfigure[]{
        \includegraphics[scale=0.45]{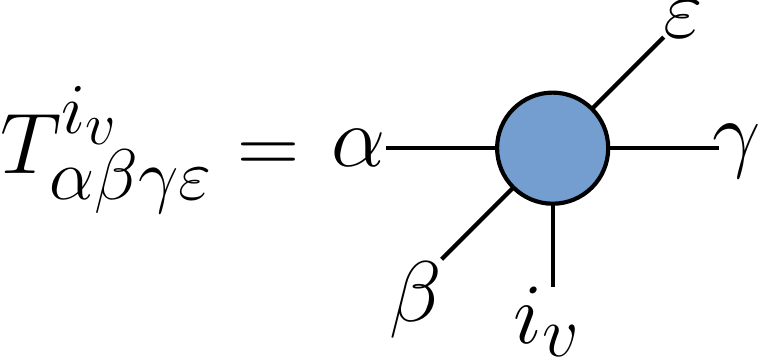}
        }
    \quad
    \subfigure[]{
        \includegraphics[scale=0.45]{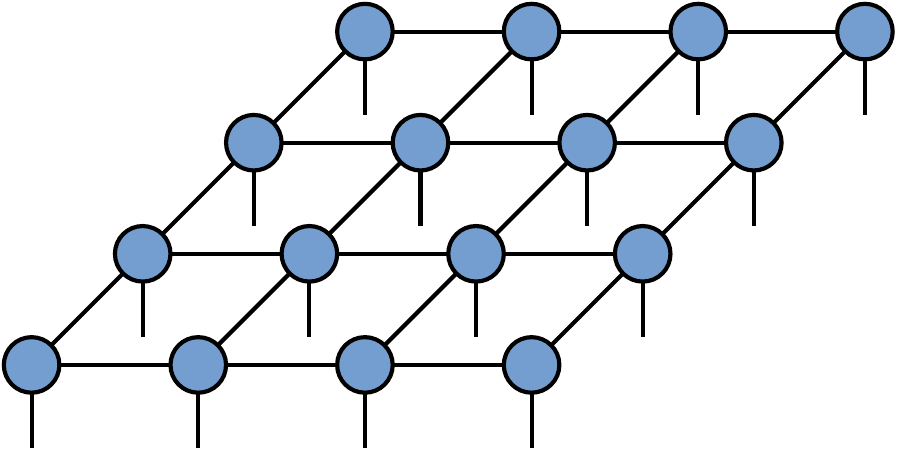}
    }
    \quad
    \caption{\label{fig:peps_tensors}(a) Rank-5 bulk tensor. (b) PEPS ansatz. Lines connecting tensors indicate contractions along ancilla indices. The entire tensor network forms a large tensor with only physical indices uncontracted, represented in the text as $\mathcal{C}(\{T^{i_v}_{\alpha\beta\gamma\varepsilon}\})$.}
\end{figure}

We consider a square lattice of dimensions $L_x \times L_y$ with open boundary conditions. Each vertex $v$ is addressed by row/column coordinates $(r,c)$ and carries a physical degree of freedom with Hilbert space dimension $d$. Next, each vertex in the bulk ($2\leq r \leq L_y-1$, $2\leq c \leq L_x-1$) has associated to it a rank-5 tensor $T^{i_v}_{\alpha \beta \gamma \varepsilon}$ where $i_v$ is referred to as a ``physical index" ($0\leq i_v \leq d-1$) and $\alpha,\beta,\gamma,\varepsilon$ are referred to as ``ancilla indices" $1\leq\alpha,\beta,\gamma,\varepsilon \leq\chi$ where $\chi$ is referred to as the ``bond dimension". Tensors along the edge of the lattice are similarly denoted and are generally rank-4 or rank-3 (corners). Graphically, tensors are indicated by solid shapes with lines emerging for each index as in Fig.~\ref{fig:peps_tensors} (a). Lines which connect two tensors indicate contraction of the two tensors along a particular index. Tensors are arranged and connected as in Fig.~\ref{fig:peps_tensors} (b) to form the tensor network, where the tensor at each site is contracted with those of its nearest neighbors along ancilla indices, leaving only the physical index uncontracted.

The resulting tensor network defines the wavefunction as follows:

\begin{equation}
    \ket{\psi} = \sum_{\vec{i}}\mathcal{C}(\{T^{i_v}_{\alpha\beta\gamma\varepsilon}\}) \ket{\vec{i}}
\end{equation}
where $\{T^{i_v}_{\alpha\beta\gamma\varepsilon}\}$ indicates the set of tensors associated to the vertices, and $\mathcal{C}$ indicates tensor contraction according to the network. This construction defines the unconstrained PEPS ansatz for quantum many-body wavefunctions.

We now impose so-called isometry constraints to obtain a subclass which we refer to as the 2D isoTNS ansatz. A tensor $T^{i_v}_{\alpha \beta \gamma \varepsilon}$ as in Fig.~\ref{fig:peps_tensors} (a) may be viewed as a map from a tensor product of three Hilbert spaces associated to the indices $i_v,\alpha,\beta$ of dimensions $d,\chi$, and $\chi$, respectively, to a tensor product of two Hilbert spaces associated to the indices $\gamma, \varepsilon$, both of dimension $\chi$. $T^{i_v}_{\alpha \beta \gamma \varepsilon}$ is an isometry from the Hilbert space of $i_v,\alpha,\beta$ to the Hilbert space of $\gamma, \varepsilon$ if
\begin{equation}
    (T^{i_v}_{\alpha \beta \gamma' \varepsilon'})^*T^{i_v}_{\alpha \beta \gamma \varepsilon} = \delta_{\gamma, \gamma'}\delta_{\varepsilon, \varepsilon'}
\end{equation}
where repeated indices are implicitly summed. Graphically, this condition is denoted by placing incoming (outgoing) arrows on the lines corresponding to the $\alpha,\beta$ ($\gamma,\varepsilon$) indices. The line for the physical index $i_v$ always has an incoming arrow which is omitted for visual clarity.

An isometric representation of the tensor network is defined by choosing a row and column of $L_x + L_y - 1$ tensors, which we refer to as an ``orthogonality hypersurface", and demanding that all tensors not on the orthogonality hypersurface are isometric with two outgoing legs pointing to the chosen row and column (row 5, column 4 in Fig.~\ref{fig:2D_isoTNS}). The entire orthogonality hypersurface is then connected to the remainder of the network only by incoming arrows. The orthogonality hypersurface itself has one tensor with only incoming arrows, which we refer to as the ``orthogonality center".

The key point is that the orthogonality hypersurface divides the lattice into four regions, such that the tensor network within each region is an isometry from the region to the ancilla space of the orthogonality hypersurface. Put differently the orthogonality hypersurface is the boundary between different isometry directions. This constraint ensures that in order to evaluate the expectation value of an observable supported on the orthogonality hypersurface, one need not explicitly or approximately contract tensors off the orthogonality hypersurface---the isometry conditions guarantee that such contractions yield the identity. The 2D isoTNS ansatz is then defined by demanding that for all choices of row and column, such an isometric representation of the wavefunction exists.

These isometry conditions generalize the canonical forms of MPS described in App.~\ref{app:isoTNS} which similarly admit efficient evaluation of observables. Indeed, the orthogonality hypersurface in 2D isoTNS itself may be viewed as an MPS with physical dimension $d\chi^2$ which can be put in mixed-canonical form to create an orthogonality center as in Fig.~\ref{fig:2D_isoTNS}. The orthogonality center can then be moved throughout the orthogonality hypersurface as in MPS. 

Of course, it is also necessary that, as with the orthogonality center in the 1D MPS ansatz, one needs to be able to move the orthogonality hypersurface around efficiently. This point is discussed further in Ref.~\cite{Zaletel19}. Unlike in the case of MPS, however, it is not yet understood what quantum phases may be represented exactly or even approximately by this ansatz. In this work, we show that string-net liquids and states which may be transformed into them by quantum circuits of finite depth can be represented exactly as an isoTNS.

\begin{figure}
\includegraphics[scale=0.60]{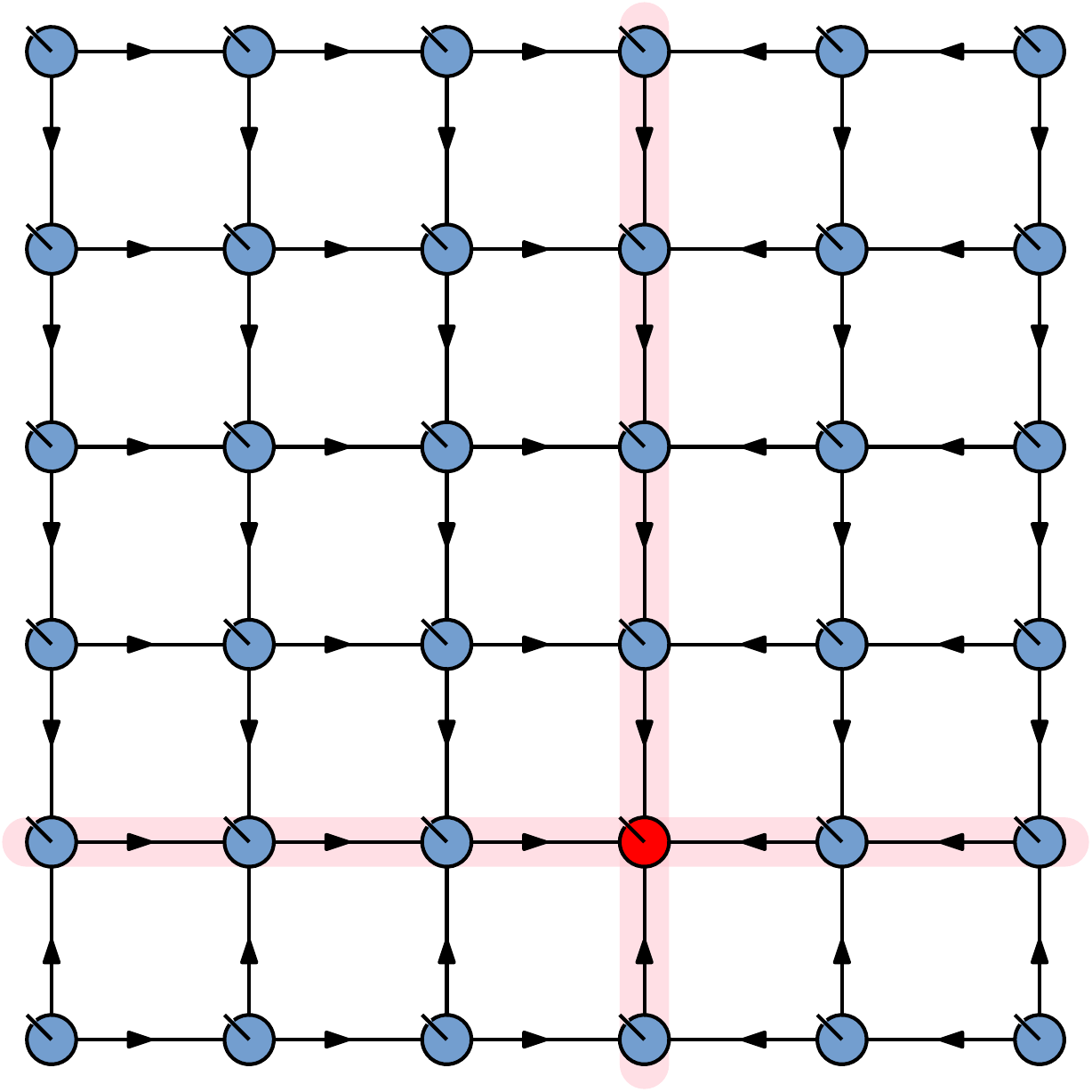}
\caption{\label{fig:2D_isoTNS}2D Isometric Tensor Network. Each tensor has an uncontracted physical index with an implicit incoming arrow. The orthogonality hypersurface is formed by the tensors in highlighted in row 5 and in column 4. The orthogonality center is the tensor marked red at the intersection of the row and column. The orthogonality hypersurface divides the lattice into four bulk regions. The isometry conditions imposed imply that the tensor formed by contracting all internal indices in each bulk subregion is an isometry from the physical indices in the bulk to the ancilla indices in the orthogonality hypersurface.}
\end{figure}

\section{\label{sec:numerics}Numerics}
We motivate the study of string-net liquids as isometric tensor networks with a numerical calculation of a well-studied example, the toric code on a square lattice \cite{Kitaev03}. 
A spin-1/2 degree of freedom, labeled by $\ket{0},\ket{1}$ (up/down spins in the $z$ basis), is assigned to each edge of the lattice. The Hamiltonian reads

\begin{equation}
\begin{split}
    H = -\sum_{s}A_s - \sum_{p}B_p \\
    A_s = \prod_{i \in s} \sigma^z_i\ ,\quad B_p = \prod_{i \in p} \sigma^x_i
\end{split}
\label{eq:toric_code_H}
\end{equation}
where $\sigma^x$ and $\sigma^z$ are  Pauli matrices. The first sum is over products of four spins on ``stars'' (four spins adjacent to a vertex) and the second sum is over products of spins on ``plaquettes'' on the lattice. This model is exactly solvable, and its ground state exhibits $\mathbb{Z}_2$ topological order. See Ref.~\cite{Kitaev03} for details regarding the ground state wavefunctions.

\begin{figure}
 \includegraphics[width=8cm]{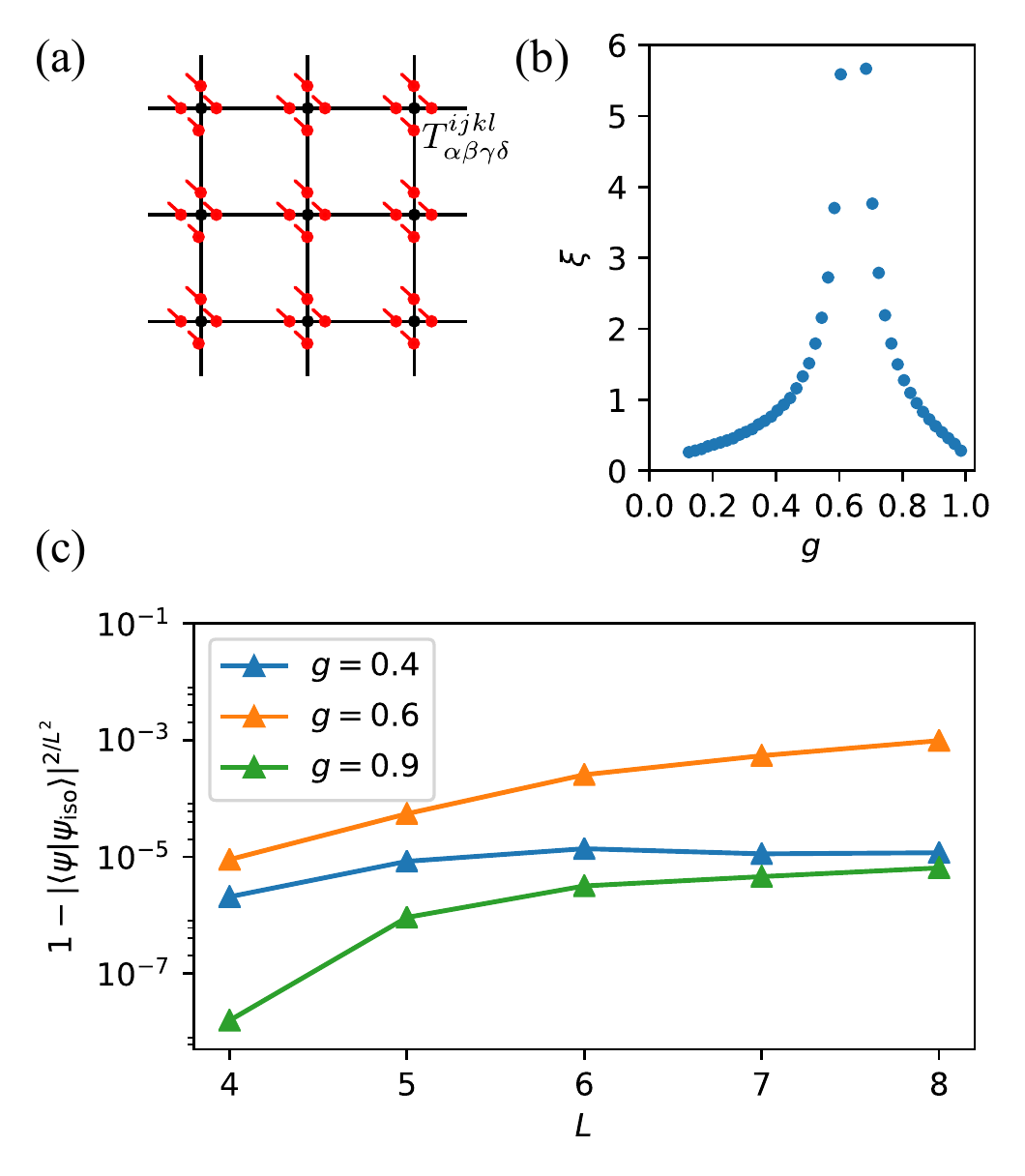} 
 \caption{\label{fig:tensor_toric} 
 (a) PEPS representation of the toric code state.
 (b) Correlation length as extracted from the transfer matrix of the boundary MPS 
\cite{Jordan08} of the perturbed toric code state. 
 (c) Error density as function of the linear system size for an isoTNS using a bond dimension $\chi=4$ and a zero-column state (an orthogonality hypersurface without physical indices used for the compression, see Ref.~\cite{Zaletel19} for details) with bond dimension $\eta=12$.
 }
\end{figure}

The ground state of the toric code can be represented exactly as a PEPS with a bond dimension $\chi=2$. To see this, we first double the local Hilbert space on each edge and take $|i\rangle \rightarrow |ii\rangle$ in the $z$ basis as illustrated in Fig.~\ref{fig:tensor_toric}(a). We then associate to each vertex the local Hilbert space of the four nearest spins. As in Fig.~\ref{fig:peps_tensors}, we assign a tensor $T$ to each vertex but now with four physical indices ($i$, $j$, $k$, $l$) corresponding to the states of the four spins and four ancilla indices ($\alpha$, $\beta$, $\gamma$, $\delta$). We define $x = i+j+k+l$, and choose $T$ such that 
\begin{eqnarray}
T^{ijkl}_{\alpha\beta\gamma\delta}=
\left\{
    \begin{array}{ll}
       \delta_{i,\alpha}\delta_{j,\beta}\delta_{k,\gamma}\delta_{l,\delta} &\mathrm{if}\ x=0 \ \mathrm{mod}\ 2\\
      0 &\mathrm{otherwise.}
    \end{array}
  \right.
  \label{eq:exact_toric_PEPS}
\end{eqnarray}
The state that results after the contraction of the PEPS is the ground state of the Hamiltonian in Eq.~\ref{eq:toric_code_H}.

In the following, we consider a generalization of the PEPS introduced above, motivated by the physics of a toric code in the presence of a magnetic field. Note that the only non-zero element in a PEPS of a completely polarized state in $z$-direction is $T^{0000}_{0000}=1$, and all other elements are zero. Although the actual ground state of the toric code in a finite field does not have a simple PEPS representation, we can construct a PEPS interpolating between the two fix point states, parameterized by a value $g$ \cite{Chen10-2}:

\begin{eqnarray}
T^{ijkl}_{\alpha\beta\gamma\delta}(g)=
\left\{
    \begin{array}{ll}
       g^{x/2}\delta_{i,\alpha}\delta_{j,\beta}\delta_{k,\gamma}\delta_{l,\delta}  & \mathrm{if}\ x=0 \ \mathrm{mod}\ 2\\
      0  &\mathrm{otherwise}
    \end{array}
  \right.
  \label{eq:perturbed_toric_PEPS}
\end{eqnarray}

This state exhibits a phase transition between the topologically trivial ($g < g_\ast$) and the $\mathbb{Z}_2$ topologically ordered  ($g > g_\ast$) phase at a critical point $g_\ast$ where the correlation length diverges as shown in Fig.~\ref{fig:tensor_toric} (b). We will now use this family of PEPS to investigate the representability of general states in terms of isoTNS. 

For a system with open boundary conditions, we choose the boundary spins (the outermost spins in Fig.~\ref{fig:tensor_toric} (a)) to be polarized in $x$-directions. This corresponds to the rough boundary condition, as defined in Ref.~\cite{Kong11}. With this boundary condition, the PEPS is in an isometric form for $g=0$ and $g=1$. Moving away from the fixed points, however, the PEPS no longer has this property. An interesting question, then, is whether the state can still be efficiently approximated by an isoTNS. 

We start from the exact $\chi = 2$ PEPS representation, and sweep right column-by-column to put the state into isometric form. We do so by using repeated variational sweeps to maximize the fidelity of the isoTNS. The results are shown in  Fig.~\ref{fig:tensor_toric} (c). Away from the critical point, the error density becomes independent of the system size once $L\gg \xi$, indicating that the PEPS can be efficiently represented even in the thermodynamic limit. Near the critical point, the error density keeps growing as the system size increases. Note that at $g=0$ and $g=1$, the algorithm finds perfect overlap which we do not show here.

These numerical results suggest that phases with topological order may be captured by the isoTNS ansatz. In the remainder of this paper, we use the framework of string-net liquids and finite-depth circuits to make this statement precise, and rigorously prove that a wide class of topologically ordered states can be represented by isoTNS ansatz.

\section{\label{sec:string_net_tensors}String-net wavefunction}
\begin{figure}
    \centering
    \includegraphics[scale=0.45]{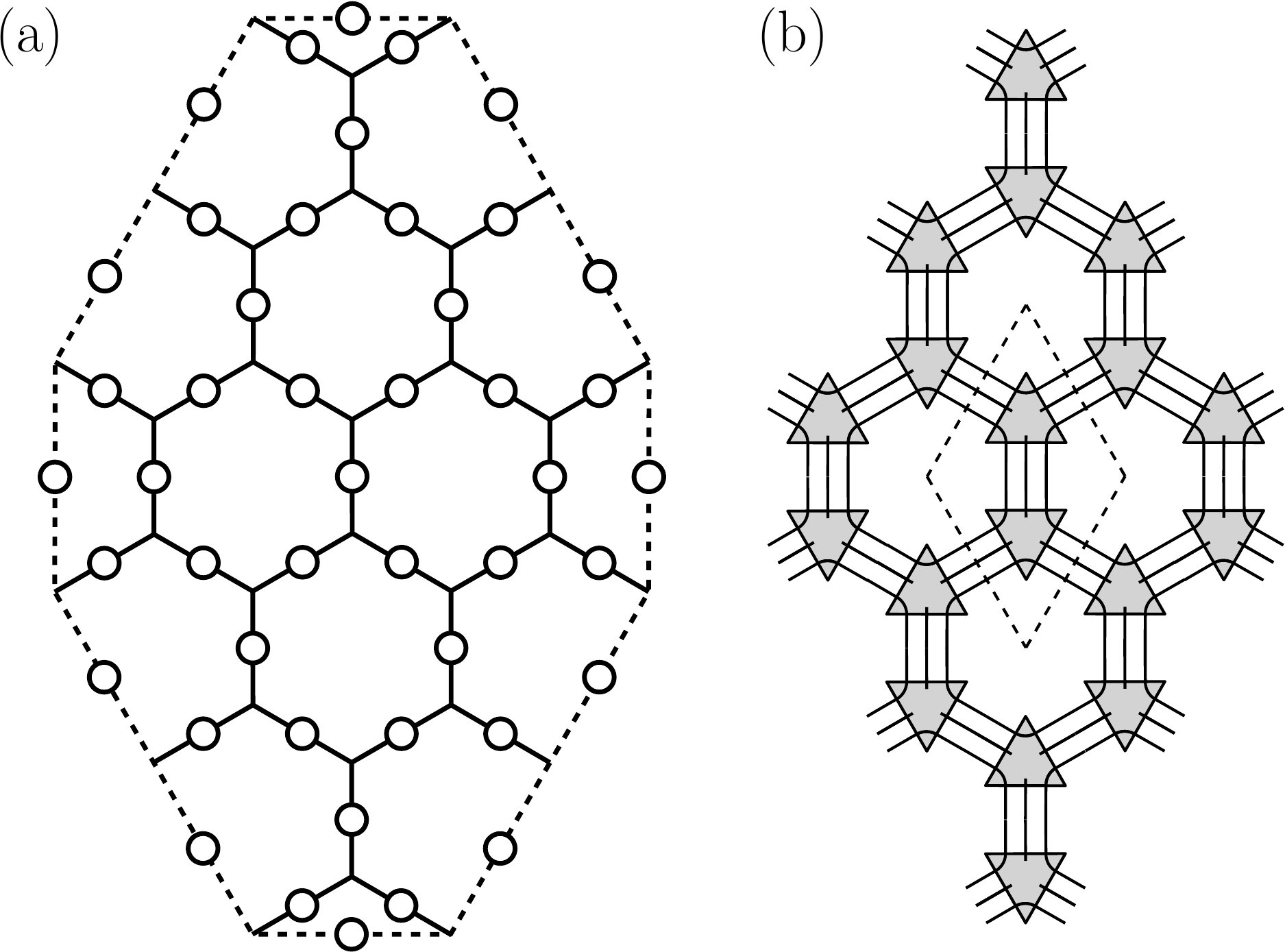}
    \caption{\label{fig:honeycomb_lattice} (a) Honeycomb lattice on which string-net liquid wavefunction is defined. Degrees of freedom are represented by empty circles. Solid line represents the bulk and dashed line represents edge of the system. The construction for the wavefunction and a PEPS representation is given in App.~\ref{app:string_net_liquid}. (b) Schematic representation of the PEPS for the same honeycomb lattice. Internal loops correspond to auxiliary degrees of freedom, while open lines correspond to physical degrees of freedom. The triangles represent tensors as defined in Eq.~\ref{eq:original_string_net_tensors} and the unit cell is denoted by dashed lines. Black and white dots, boundary conditions etc. are omitted for simplicity.}
\end{figure}

\begin{figure}
    \centering
    \includegraphics{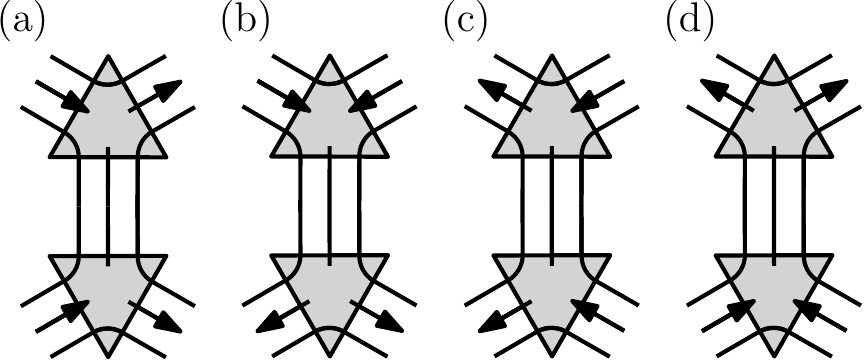}
    \caption{Schematic representation of the four isometry directions for the unit cell of a honeycomb lattice. Arrows indicate the isometry directions.}
    \label{fig:unit_cell_isometry}
\end{figure}

String-net models are exactly solvable lattice models, whose ground state wavefunctions can be thought of as fixed-point wavefunctions of topological phases of matter. They provide concrete lattice realizations of various topological phases, and are therefore suitable for understanding the ability of isoTNS to represent topological phases.

A 2D string-net wavefunction is defined on a trivalent lattice where degrees of freedom live on the edges of the lattice (Fig.~\ref{fig:honeycomb_lattice} (a)). The local Hilbert space is spanned by $N+1$ basis states $i={0, 1, ..., N}$ that correspond to ``string types'', each associated with a positive number $d_i > 0$ known as that string type's quantum dimension. The strings are allowed to ``branch'' according to ``branching rules'' or ``fusion constraints'' $\delta_{ijk}$ which is 1 if string types $ijk$ can meet at a vertex and 0 otherwise. The ground state wavefunction is characterized by a set of graphical rules relating different string configurations. A central object in that graphical rule is a six-index tensor called the $F$-symbol, denoted $F^{ijm}_{kln}$, which satisfies certain compatibility requirements. Given $F$-symbols and string types, we can construct an exactly solvable projector Hamiltonian, the ground states of which satisfy the graphical rules (see App.~\ref{app:string_net_liquid} for a more complete treatment of string-net liquid and $F$-symbols).

It has been shown that the string-net ground state wavefunctions have an exact tensor network representation defined in terms of $F$-symbols \cite{Gu2009, Buerschaper2009}. The tensor network is easiest to depict after doubling the degrees of freedom on edges such that each vertex has three physical degrees of freedom associated to it. The network is depicted schematically in Fig.~\ref{fig:honeycomb_lattice} (b) where tensors on upper (lower) sublattice are depicted by solid triangles pointing up (down).

The precise tensor network depends on a choice of branching structure, but one particularly simple representation can be obtained by using only two types of tensors in the bulk of the honeycomb lattice, such that all bulk tensors on a given sublattice are the same. The bulk tensor for the upper sublattice sites is given by the following:

\begin{equation}
    \vcenter{\hbox{\includegraphics[width=0.35\linewidth]{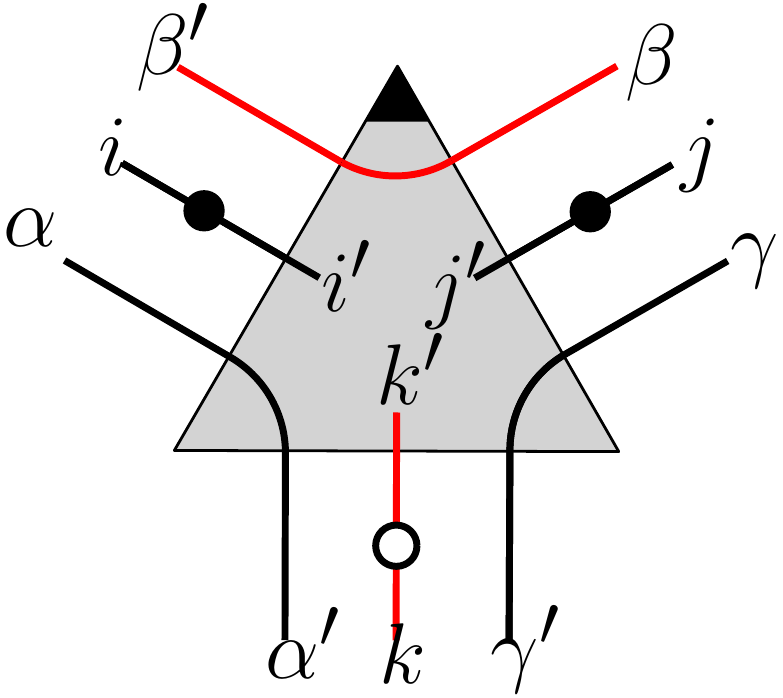}}} = 
    \begin{array}{l}    \delta_{\alpha\alpha'}\delta_{\beta\beta'}\delta_{\gamma\gamma'} \\
    \delta_{ii'}\delta_{jj'}\delta_{kk'}
    \end{array}\sqrt{\frac{v_iv_j}{v_k}}F^{i\alpha^*\beta}_{\gamma j k}
    \label{eq:full_tensor_notation}
\end{equation}
where $v_i = (F^{ii^*0}_{ii^*0})^{-1/2}$ and we have introduced graphical notations for Kronecker delta and factors of $\sqrt{v_i}$ on the LHS (see Fig.~\ref{fig:graphical_notations}). The $v_i$'s satisfy $v_i \times v_i^* = d_i$, and in the remainder of the text we will loosely refer to these as factors of quantum dimension. Each leg of the tensor takes $N+1$ possible values, corresponding to the number of string types. The indices $i',j',k'$ label the three physical degrees of freedom, while the remaining indices label the ancilla legs. The grey triangle represents the $F$-symbol and the red legs indicate unitary legs as discussed in Eq.~\ref{eq:F_unitary_in_subspace}.

We note that even though the tensor has twelve indices, for brevity we sometimes write these tensors with six indices, making the Kronecker deltas, like those in Eq.~\ref{eq:full_tensor_notation}, implicit. With this convention, the tensors for two sublattices are given by the following:

\begin{equation}
\begin{alignedat}{3}
    \vcenter{\hbox{\includegraphics[width = 0.4\linewidth]{figures/up_triangle.pdf}}} &= \sqrt{\frac{v_i v_j}{v_k}} F^{i\alpha^* \beta}_{\gamma j k}\ &\text{(upper)}\\
    \vcenter{\hbox{\includegraphics[width=0.4\linewidth]{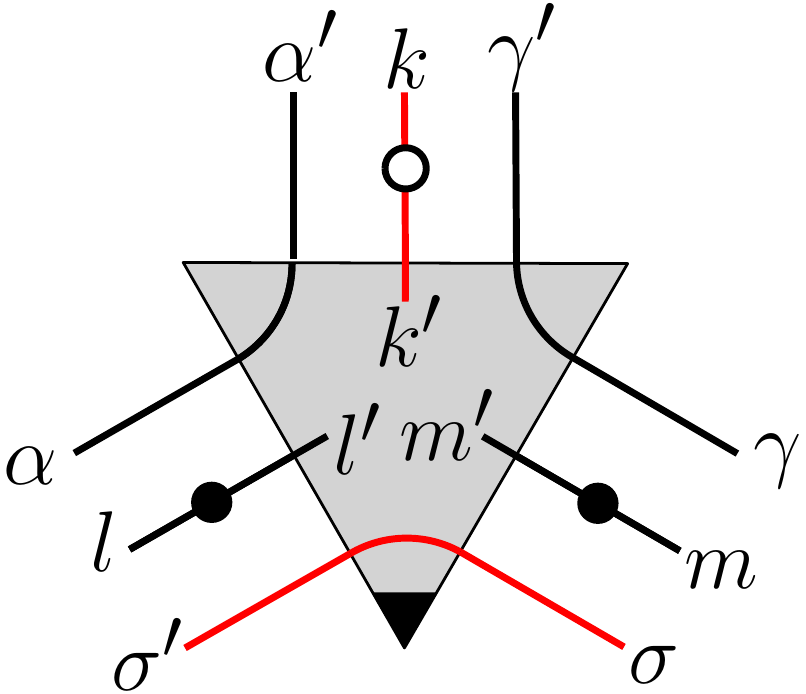}}}  &=  \sqrt{\frac{v_l v_m}{v_k}} (F^{l \alpha^* \sigma}_{\gamma m k})^*\ &\text{(lower)}
\end{alignedat}
    \label{eq:original_string_net_tensors}    
\end{equation}
This gives us the tensors in the bulk. However, edge tensors, the outermost tensors directly connected to the open indices on the edges and corners of the system, take different forms. Further details on the tensor network construction of string-net ground states, including how to construct the edge tensors, are given in App.~\ref{app:string_net_liquid}.

Although we have defined isoTNS ansatz only for square lattice in Sec.~\ref{sec:isoTNS_intro}, the honeycomb lattice can be thought of as a square lattice if we consider the unit cell as shown in Fig.~\ref{fig:honeycomb_lattice} (b). Each unit cell has four neighboring unit cells, and therefore the honeycomb lattice can be mapped to a square lattice. The tensor for the unit cell (``unit cell tensor'') can be constructed by putting the tensors for upper and lower sublattices together (Fig.~\ref{fig:five_tensors} (e)). As in Fig.~\ref{fig:2D_isoTNS}, there are four isometry directions for the unit cell tensor, which we label (a-d) in Fig.~\ref{fig:unit_cell_isometry}. We will use this labeling for isometry directions throughout the main text (Figs.~\ref{fig:five_tensors}, \ref{fig:four_isometries},\ref{fig:isometric_and_orthogonal}).

Our present work shows that the string-net unit cell tensors can be put in an isometric form by exploiting tetrahedral symmetries of the $F$-symbols and gauge redundancies in the tensor network description.

\begin{figure*}
    \centering
    \includegraphics[width=\linewidth]{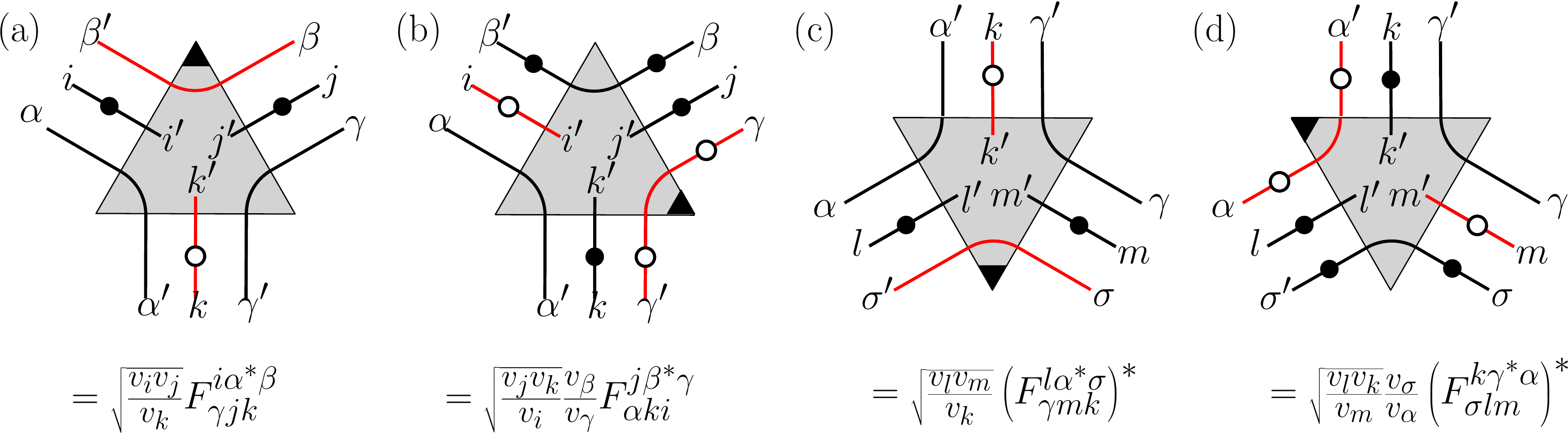}
    \caption{Graphical notation for string-net tensors. The Kronecker deltas that come with solid lines are suppressed. Black triangles shows the unitary direction and red lines indicate unitary legs. Tensors (a) and (c) represent the original ``first form'' tensors in Eq.~\ref{eq:original_string_net_tensors}. Tensors (b) and (d) are the ``second form'' tensors, obtained using Eq.~\ref{eq:convenient_tetrahedral_constraints}.}
    \label{fig:tensor_definition}
\end{figure*}

\begin{figure}
    \centering
    \includegraphics[width = \linewidth]{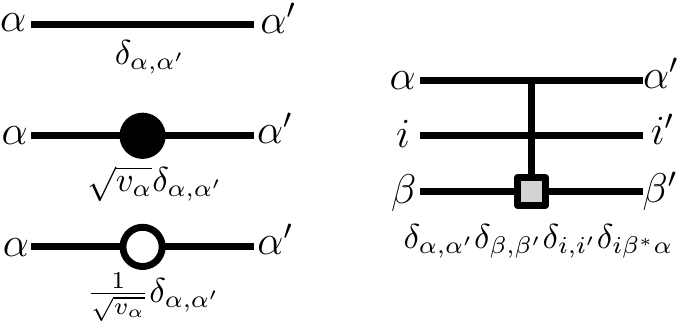}
    \caption{Some graphical notations used throughout the text. The light gray box for the fusion constraint indicates the other two indices fuse to the gray index.}
    \label{fig:graphical_notations}
\end{figure}

\section{\label{sec:string_isometry}Transforming string-net tensors into isometric tensors}

In this section, we construct an isometric form for the string-net tensors. The main strategy, which is outlined by a sequence of figures, is as follows. We will use the ability to rewrite the tensors using tetrahedral symmetry (Fig.~\ref{fig:tensor_definition}) as well as the ability to perform gauge transformations in order to put the tensors in a desired form (Figs.~\ref{fig:quantum_dims} and~\ref{fig:five_tensors}). Then, we will ``strip off'' some of the fusion constraints attached to $F$-symbols (Figs.~\ref{fig:full_rank_tensors} and~\ref{fig:four_isometries}) and place them all at the orthogonality center (Fig.~\ref{fig:isometric_and_orthogonal}). With the help of a graphical notation, we can show these tensors are isometric and equivalent to the original tensor network (Figs.~\ref{fig:fusion_moves} and~\ref{fig:isometry_proof}). Finally, in App.~\ref{app:entanglement_spectrum}, we demonstrate that these isometric tensors provide an easy method for calculating the entanglement spectrum.

Note that the proof in main text covers string-net liquid as defined in \cite{Levin05}. Its generalizations as defined in \cite{Lin14} need a different treatment, as discussed in App.~\ref{app:abelian_isometry}.

\subsection{\label{sec:f_identities}\texorpdfstring{$F$}{F}-symbol identities}
We briefly summarize some properties of the $F$-symbols that are necessary for constructing the new tensors and proving their isometry. We show in App.~\ref{app:string_net_liquid} that these properties are satisfied by isotopy-invariant string-nets. First, $F$-symbols satisfy so-called ``branching rules'' or ``fusion constraints'':

\begin{equation}
    F^{i\alpha^*\beta}_{\gamma j k} = F^{i\alpha^*\beta}_{\gamma j k} \delta_{i\alpha^*\beta} \delta_{j\beta^*\gamma} \delta_{k\gamma^*\alpha} \delta_{ijk}
    \label{eq:fusion_constraints}
\end{equation}
Next, as a result of tetrahedral symmetries, the $F$-symbols satisfy the following equalities:

\begin{equation}
    \sqrt{\frac{v_i v_j }{v_k}}F^{i\alpha^*\beta}_{\gamma j k} = \sqrt{\frac{v_k v_i}{v_j}}\frac{v_\beta}{v_\alpha}F^{k\gamma^*\alpha}_{\beta i j} = \sqrt{\frac{v_j v_k}{v_i}}\frac{v_\beta}{v_\gamma}F^{j\beta^*\gamma}_{\alpha k i}
    \label{eq:convenient_tetrahedral_constraints}
\end{equation}
This allows us to write the expression in Eq.~\ref{eq:original_string_net_tensors} in different ways. Graphically, we observe Fig.~\ref{fig:tensor_definition} (a) and (c) are the same tensors as Fig.~\ref{fig:tensor_definition} (b) and (d) thanks to Eq.~\ref{eq:convenient_tetrahedral_constraints}. In the remainder, we will refer to the former form of tensors as ``the first form'' and the latter forms of tensors as ``the second form''. The second form has different factors of quantum dimensions, and its indices for $F$-symbols are permuted from that for the first form.

Next, the string-net $F$-symbols are unitary matrices when restricted to the subspace satisfying fusion constraints: 

\begin{equation}
    \sum_{k=0}^N (F^{i\alpha^*\beta}_{\gamma jk})^*F^{i\alpha^*\beta'}_{\gamma jk} = \delta_{\beta,\beta'}\delta_{i\alpha^*\beta}\delta_{j\beta^*\gamma}
    \label{eq:F_unitary_in_subspace}
\end{equation}
We call the indices $\beta, k$ ``unitary indices'' and corresponding legs for the tensor ``unitary legs''. Unitary legs are indicated by red in graphical notation (Fig.~\ref{fig:tensor_definition} (a), (c)). Note that the unitary legs change when we rewrite the tensors using tetrahedral symmetry (Fig,~\ref{fig:tensor_definition} (b), (d)). This unitarity property will be the starting point of our construction.

Finally, the quantum dimensions are given by $d_i = |v_i|^2$ and satisfy the following relation: 

\begin{equation}
    \sum_k \delta_{ij k^*} d_k = d_i d_j
    \label{eq:branching_frobenius_perron}
\end{equation}

\subsection{\label{sec:quantum_dims}Gauge transformations}

The unit cell tensor given in Fig.~\ref{fig:five_tensors} (e), is not a unique choice, since we can perform gauge transformations. A particular choice of gauge transformations is multiplication by $\sqrt{v_\alpha}$ (or its inverse). We can use this to change factors of $\sqrt{v_\alpha}$ for each unit cell tensor. Graphically, this can be represented by moving black (white) dots along solid lines and canceling them with each other when necessary. We can do a similar graphical manipulation withing a unit cell to move black (white) dots on solid lines to simplify the tensor. Since each solid line represents a Kronecker delta, this does not change the unit cell tensor. We utilize these operations, as well as tetrahedral symmetry (Eq.~\ref{eq:convenient_tetrahedral_constraints}) to change the unit cell tensor as in Fig.~\ref{fig:quantum_dims}.
We subsequently observe that one of the Kronecker deltas on each inner loop is redundant, and cut one leg at the top right/bottom left (indicated by dashed line). This ensures that the value of the top-right/bottom-left index does not depend on the value of bottom-right index. This results in the tensors in Fig.~\ref{fig:five_tensors}.

\begin{figure}
    \centering
    \includegraphics[width=\linewidth]{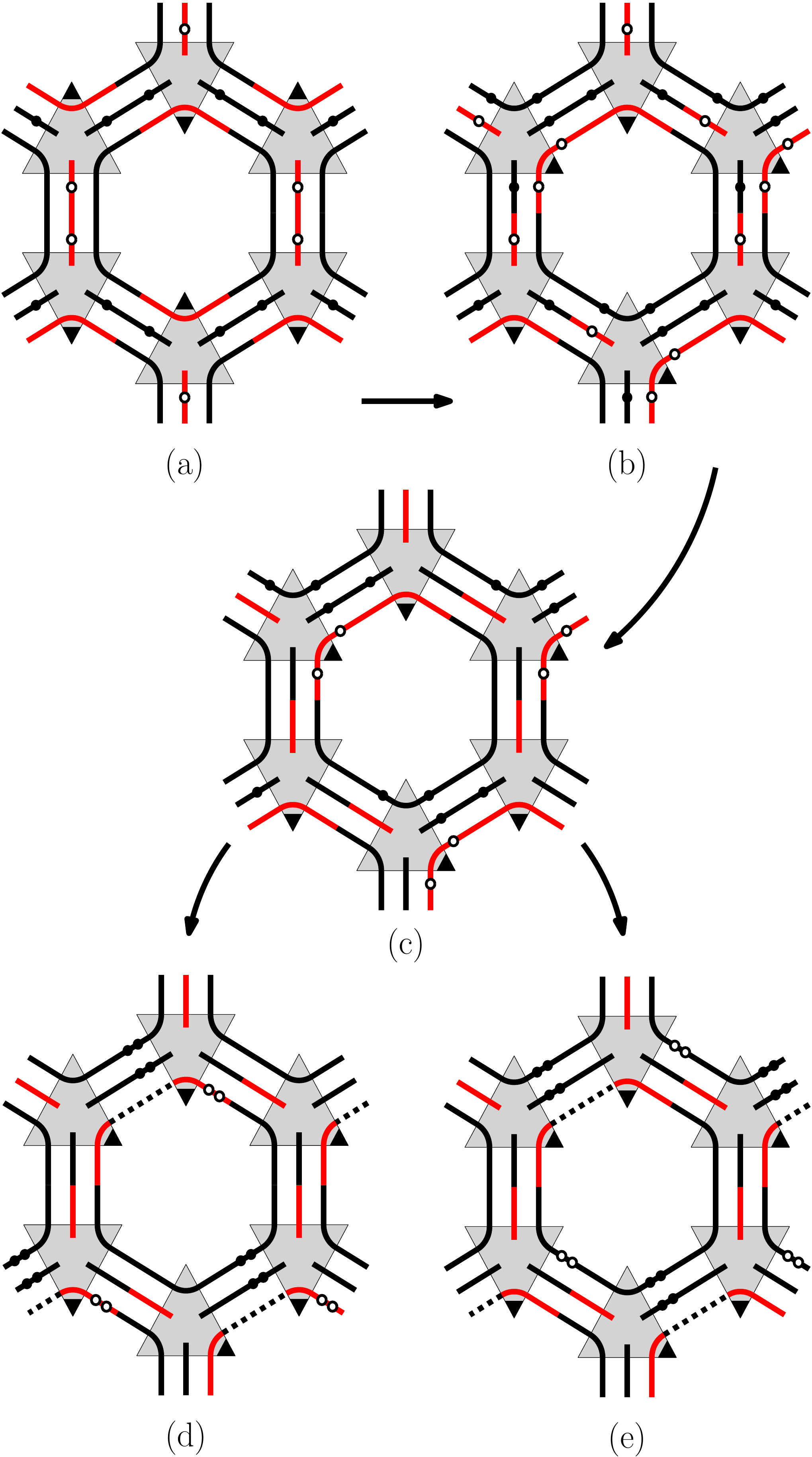}
    \caption{\label{fig:quantum_dims}Illustration of different forms of the PEPS tensor and application of gauge transformation. Arrows indicate the order of the steps. (a) The PEPS in terms of the original tensors in Eq.~\ref{eq:original_string_net_tensors}. (b) The PEPS in terms of one tensor in the second form (b) and one tensor in the first form (c) per unit cell from Fig.~\ref{fig:tensor_definition}. (c) Cancellation of factors of $\sqrt{v_\alpha}$. (d) Moving factors of $\sqrt{v_\alpha}$ and removal of redundant Kronecker delta on a loop to produce tensor (a) in Fig.~\ref{fig:five_tensors}. (e) Same procedure to produce tensor (b) in Fig.~\ref{fig:five_tensors}. Dashed lines are visual aid for removed legs.}
\end{figure}

\begin{figure}
    \centering
    \includegraphics[width = \linewidth, draft=false]{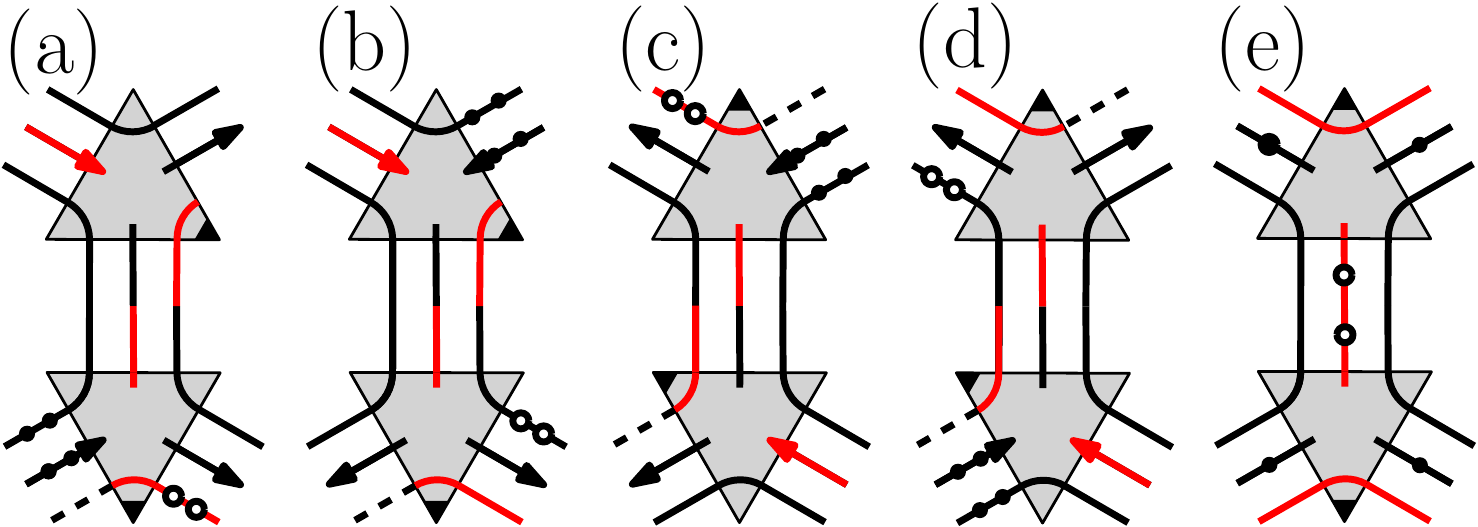}
    \caption{Five tensors that can be transformed to each other by gauge transformations. Examples of these transformations are depicted in Fig.~\ref{fig:quantum_dims}. (a-d) Tensors after gauge transformations and removing a redundant Kronecker delta. See Fig.~\ref{fig:four_isometries} for their isometric counterparts. (e) The original tensor, as in Fig.~\ref{fig:quantum_dims}.}
    \label{fig:five_tensors}
\end{figure}

\subsection{\label{sec:orthogonality_center}Isometry and Orthogonality Center}
Four tensors (a-d) in Fig.~\ref{fig:five_tensors} are almost in an isometric form for isometry directions indicated by arrows. To see this, we perform the tensor contraction and use the unitarity in a subspace (Eq.~\ref{eq:F_unitary_in_subspace}) to evaluate the resulting tensor. The detail of this procedure is omitted, since it is similar to the procedure for contraction of the isometric counterpart, which we will describe in detail in Fig.~\ref{fig:isometry_proof}.

Due to the fusion constraints (Eq.~\ref{eq:fusion_constraints}), the result of the contraction is not an identity, but rather a projector into a subspace that satisfies those fusion constraints. In order to transform this to an isometric form, we need to remove those constraints. We define the ``$A$-symbols" by ``stripping off" the fusion constraints from the $F$-symbols. This can be achieved by lifting the unitarity of $F$-symbols in a constrained subspace (Eq.~\ref{eq:F_unitary_in_subspace}) to the full space. We therefore demand the six-index $A$-symbols satisfy the following conditions:
\begin{equation}
    \begin{aligned}
        F^{ijm}_{kln} = A^{ijm}_{kln} \delta_{ijm} \delta_{klm^*}&\\
        \sum_{n=0}^N (A^{ijm}_{kln})^*A^{ijm'}_{kln} = \delta_{m,m'}&
    \end{aligned}
    \label{eq:full_rank_tensors}
\end{equation}

We show in App.~\ref{app:full_rank_tensors} that for each set of $F$-symbols, we can always construct such $A$-symbols (this construction is not unique). To distinguish between tensors defined in terms of $F$-symbols and those defined in terms of $A$-symbols, we introduce the following terminology. We call the counterpart to the 12-index $F$-symbol tensor the 12-index $A$-symbol tensor depicted in Fig.~\ref{fig:full_rank_tensors} (a). The PEPS tensors built from $F$-symbols will be referred to as ``constrained" PEPS tensors (Fig.~\ref{fig:five_tensors} (a-d)) while those built from $A$-symbols will be referred to as ``full-rank" PEPS tensors (Fig.~\ref{fig:four_isometries}).

The relations in Eq.~\ref{eq:full_rank_tensors}, drawn in terms of the $A$-symbol tensor, are shown graphically in Fig.~\ref{fig:full_rank_tensors} (b), (c), where we again use red to indicate unitary indices. We can graphically see the $A$-symbol tensor is an isometry from two incoming legs to one outgoing leg (Fig.~\ref{fig:full_rank_tensors} (c)).

In terms of these newly defined tensors, the isometric forms for string-net tensors outside of the orthogonality hypersurface are given as in Fig.~\ref{fig:four_isometries}. The orthogonality hypersurface is shown in Fig.~\ref{fig:isometric_and_orthogonal}. As in Fig.~\ref{fig:2D_isoTNS}, tensors outside of the orthogonality hypersurface are isometry from two incoming legs to two outgoing legs, while tensors inside the orthogonality hypersurface are isometry from three incoming legs to one outgoing leg.

\begin{figure}
    \centering
    \includegraphics[width = \linewidth]{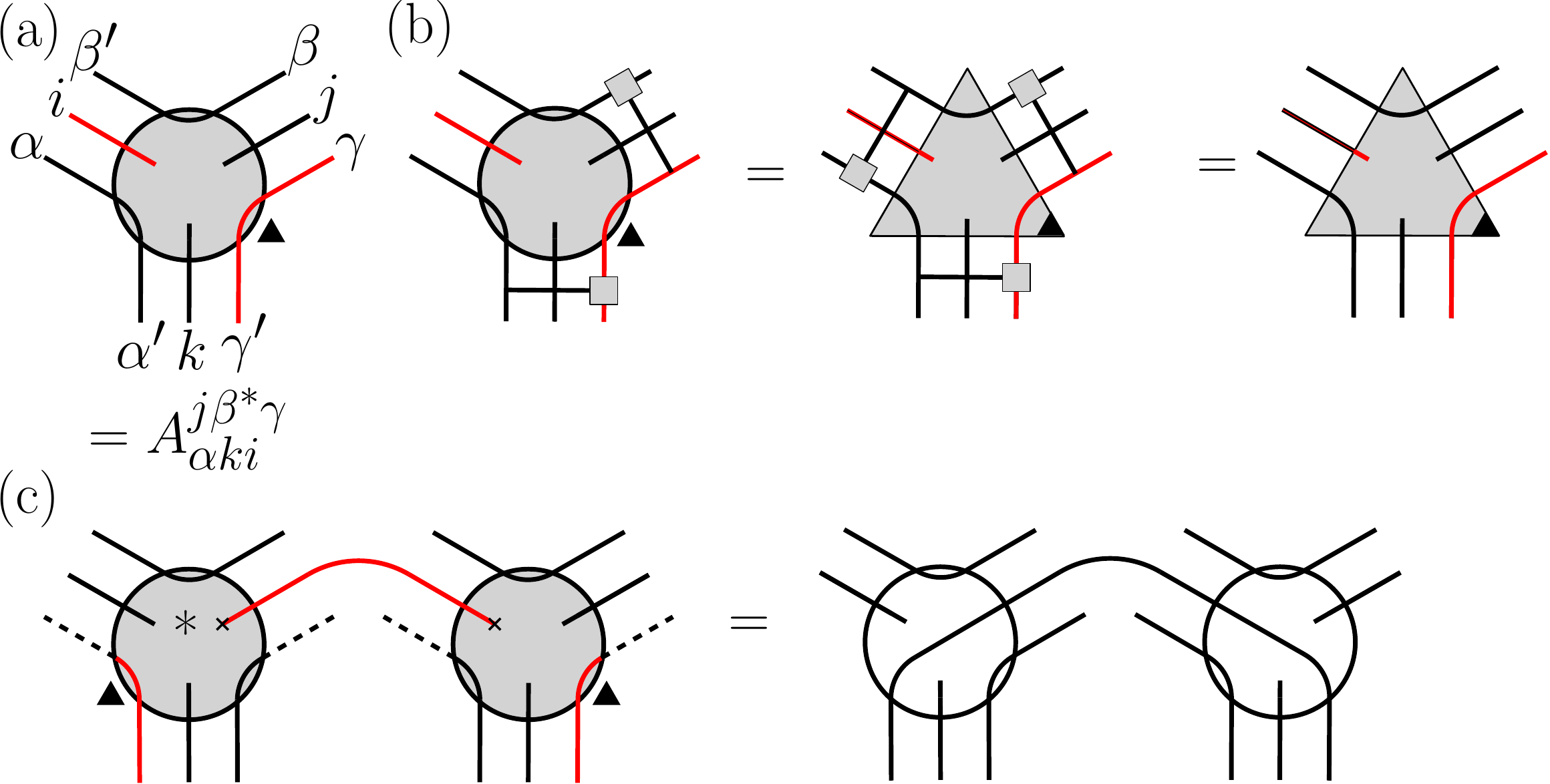}
    \caption{Graphical notations for the $A$-symbol tensor. (a) Graphical notation for the $A$-symbol tensor. Red legs label unitary indices. The black triangle is a reminder that this tensor is an isometry from top-left legs (labeled $\alpha, i, \beta'$) and physical legs (labeling suppressed) to the top-right legs ($\beta, j, \gamma)$ and bottom legs ($\alpha', k, \gamma'$). (b) When acted by fusion constraints on outgoing legs, the $A$-symbol tensor become the $F$-symbol tensor. The second equality follows from Eq.~\ref{eq:fusion_constraints}. (c) Graphical representation of unitarity. When the $A$-symbol tensor is contracted with its complex conjugate on the incoming red unitary leg, it gives a Kronecker delta for the other unitary leg. Empty circles are there for visual clarity, but they do not have any tensor content. By connecting the remaining top-left legs and physical legs, it is easy to see this implies isometry. Note that in order to contract the unitary leg, we need to contract one ancilla index (connected by a line) and one physical index (represented by crosses).}
    \label{fig:full_rank_tensors}
\end{figure}

\begin{figure}
    \centering
    \includegraphics[width=\linewidth]{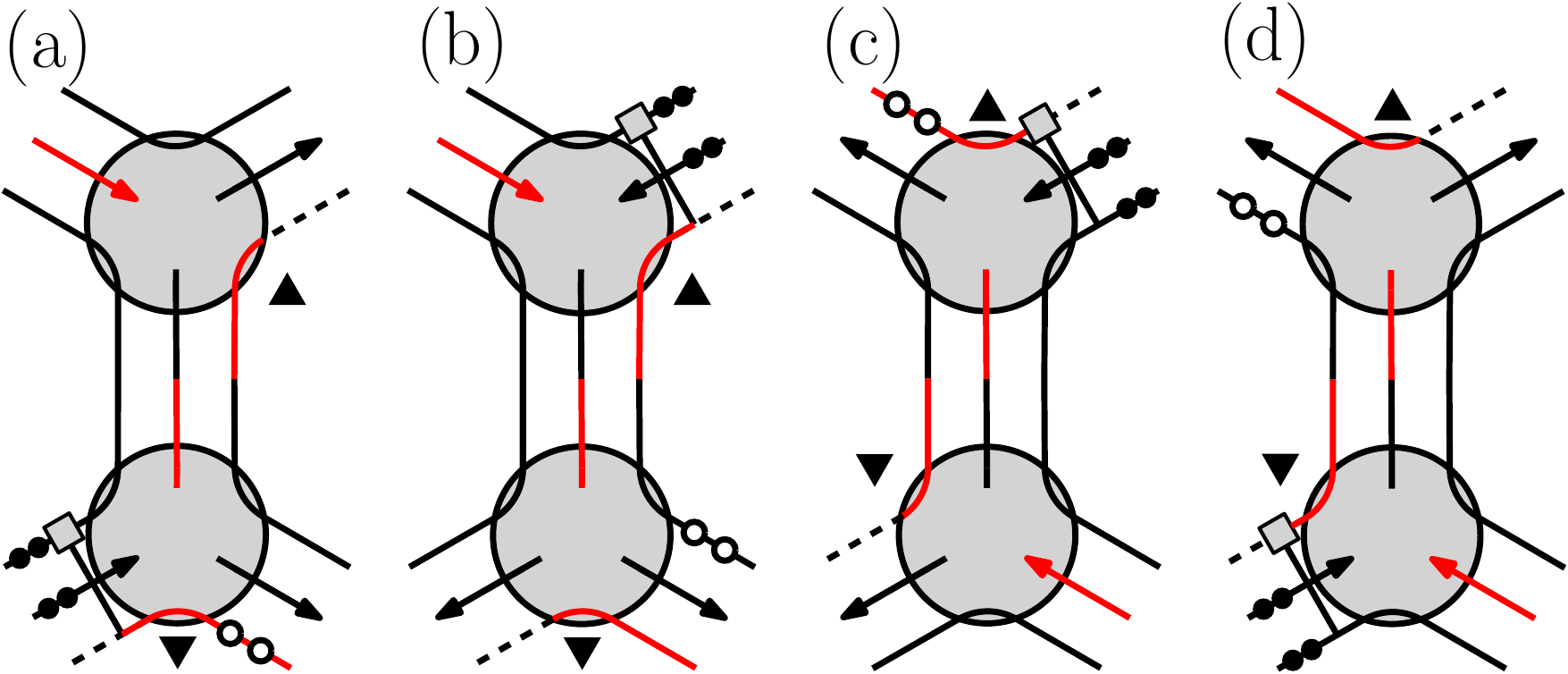}
    \caption{Four isometric tensors built from full-rank tensors. Isometry directions for the tensors are indicated by arrows. Dashed lines serve as an visual aid to indicate removed legs.  These isometric tensors (a)-(d) are obtained by ``stripping off'' fusion constaints from tensors (a)-(d) in Fig.~\ref{fig:five_tensors}.}
    \label{fig:four_isometries}
\end{figure}

\begin{figure}
    \centering
    \includegraphics[draft=false,width=\linewidth]{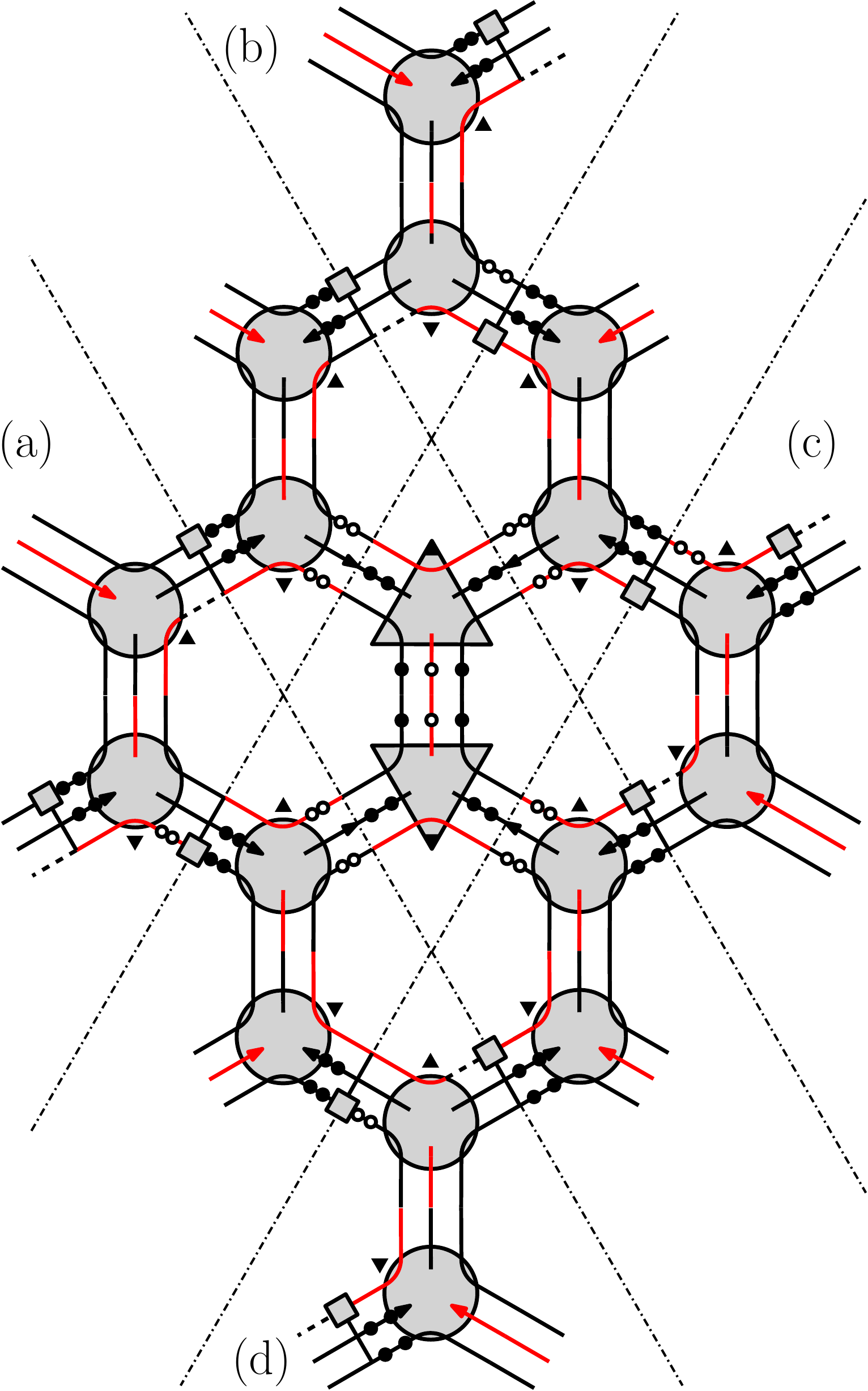}
    \caption{Orthogonality hypersurface and orthogonality center. Dot-dashed lines separate the orthogonality hypersurface from the rest of the network. The four bulk regions are labeled (a-d), corresponding to four isometric tensors defined in Fig.~\ref{fig:four_isometries}. The fusion constraints on the dot-dashed lines live inside the orthogonality hypersurface. Tensors outside of the orthogonality hypersurface are isometry from two incoming legs to two outgoing legs, while tensors inside the orthogonality hypersurface, except for the orthogonality center, are isometry from three incoming legs to one outgoing leg. These tensors are correct even at the outermost edge i.e. the edge tensors carry no extra factors of $v_i$.}
    \label{fig:isometric_and_orthogonal}
\end{figure}

In order to see that this tensor network, defined in terms of full-rank PEPS tensors, is the same as the original tensor network, observe that full-rank PEPS tensors become the constrained PEPS tensors when they are contracted with constrained tensors on outgoing legs (Fig.~\ref{fig:fusion_moves} (a),(c)). We can start this contraction at the orthogonality center and then work our way across the entire network to convert all tensors to the original, constrained tensors (Fig.~\ref{fig:fusion_moves} (e)). 

\begin{figure*}
    \centering
    \includegraphics[width = \linewidth, draft = false]{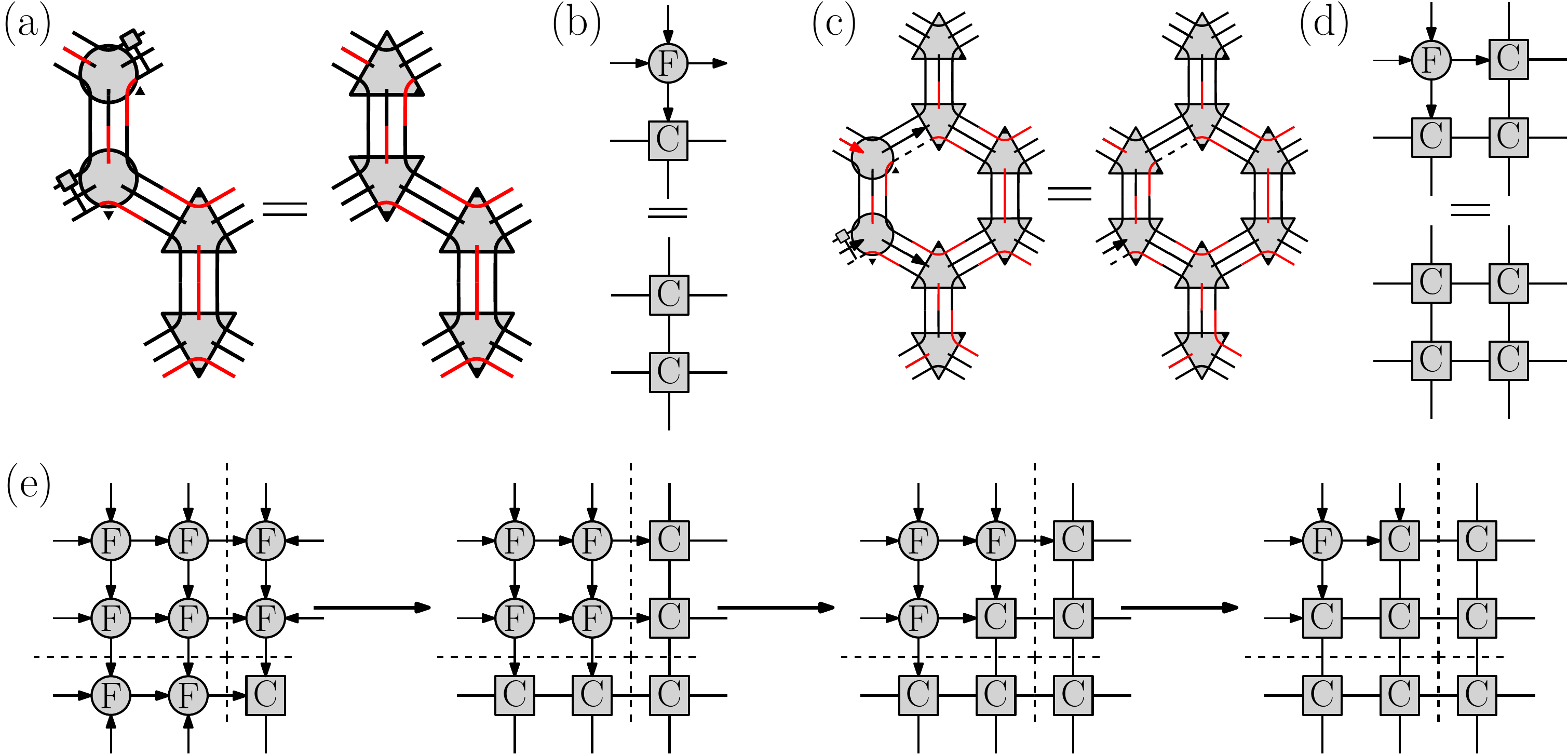}
    \caption{Fusion constraints at the orthogonality center and orthogonality hypersurface propagate through the network via Eq.~\ref{eq:full_rank_tensors}. See also Fig.~\ref{fig:full_rank_tensors} (b) for graphical notation. We omit factors of quantum dimensions for simplicity. (a) Full-rank PEPS tensor in the orthogonality hypersurface becomes constrained PEPS tensor when it is attached to the orthogonality center. (b) Schematic representation of (a) where the square tensors labeled by $F$ ($C$) represent the full-rank (constrained) PEPS tensors.  (c) Full-rank PEPS tensors outside of the orthogonality hypersurface become constrained tensors when attached to the orthogonality hypersurface.  (d) Schematic representation of (c). (e) Using the schematic notation in (b) and (d), we can see how constrained PEPS tensors `propagate' from the orthogonality center all the way to the edge.}
    \label{fig:fusion_moves}
\end{figure*}

\subsection{\label{sec:isometry_proof}Proof of isometry}
Let us prove the four tensors defined in Fig.~\ref{fig:four_isometries} are indeed isometries.
We will focus on one isometry direction, namely tensor (a) in Fig.~\ref{fig:four_isometries}. The proof for other directions and tensors inside the orthogonality hypersurface are done similarly. The main strategy is to use the isometry of the $A$-symbol tensor as described in Fig.~\ref{fig:full_rank_tensors} (c).  The contraction process is best described using graphical methods, as described in  Fig.~\ref{fig:isometry_proof}.

The only nontrivial part of the contraction is the evaluation of the inner loops that appear in Fig.~\ref{fig:isometry_proof} (c), which boils down to evaluating the following sum where $\sigma$ labels the bottom bond and $\alpha, l$ label the inner loops as in the figure:
 
\begin{equation}
    \sum_{\alpha, l} d_\alpha d_l \delta_{\alpha l \sigma^*} = \left(\sum_\alpha d_\alpha^2\right) d_\sigma = \mathcal{D} d_\sigma
\end{equation}

\noindent where we used Eq.~\ref{eq:branching_frobenius_perron} and defined $\mathcal{D} = \sum_k d_k^2$, the total quantum dimension. This factor of $d_\sigma$ cancels with the inverse factors of $d_\sigma$ attached to outgoing legs, and in the end we get an identity for outgoing legs. This completes the proof that the tensors are indeed isometries.

\begin{figure}
    \centering
    \includegraphics[width = \linewidth, draft=false]{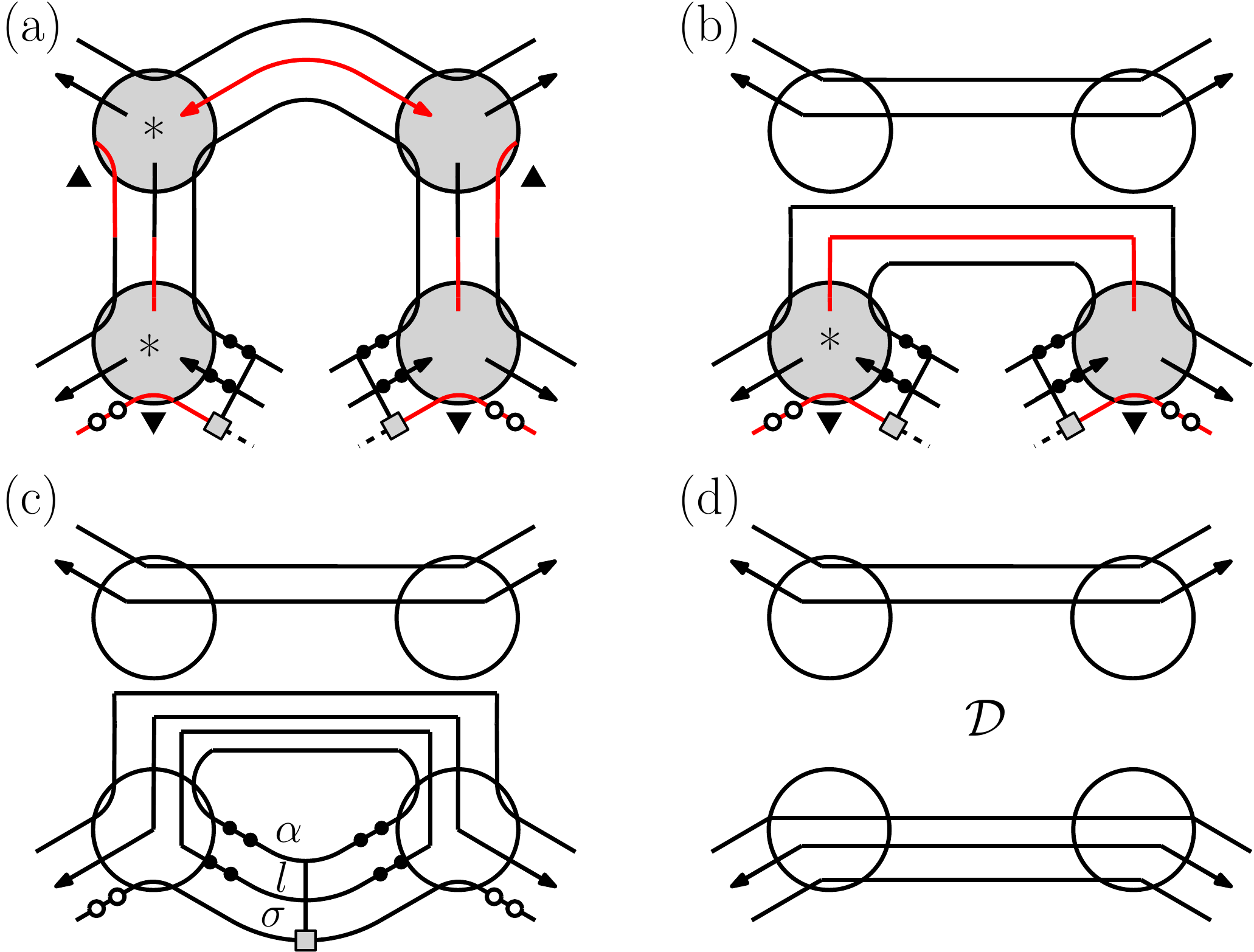}
    \caption{We contract the isometric tensor and its complex conjugate to show it is indeed isometry. We first use unitarity of the top tensor by contracting red unitary legs (a). After contracting physical legs, this will give a Kronecker delta for the top two legs (b). Contraction of the red unitary legs for the bottom tensor likewise produces Kronecker deltas (c). Finally, we evaluate the inner loops in (c) that carry $v_\alpha \times v_\alpha^* = d_\alpha$, as in the main text. The result is figure (d), which is an identity tensor times $\mathcal{D} = \sum_k d^2_k$.}
    \label{fig:isometry_proof}
\end{figure}

\section{\label{sec:finite_depth}Finite depth circuit}
In Sec.~\ref{sec:numerics}, it was found that the PEPS obtained by perturbing the exact toric code tensors could also be approximated by an isoTNS of finite bond dimension. We now show that the existence of an isoTNS representations of wavefunctions with long-range entanglement is not just a feature of the special fixed-point string-net liquids. Indeed, any state which may be transformed into a fixed-point by a local unitary quantum circuit of finite depth may be represented exactly by an isoTNS of finite bond dimension. Transformation by a quantum circuit is commonly regarded as a valid approximation to adiabatic evolution under a local, gapped Hamiltonian. We make  a version of this  statement  precise in Theorem~\ref{thm:adiabatic_evolution_fdc} in App.~\ref{app:short_depth_circuit} by carefully translating continuous quasi-adiabatic evolution as defined in Ref.~\cite{Hastings05} into a circuit of finite-depth and examining local observables.

First, we define some terminology pertaining to features of a quantum circuit. We then use Figs.~\ref{fig:2D_circuit} and~\ref{fig:graining_moves} to present a sequence of transformations used to preserve the isometric form of the tensor network after unitary operations are applied to the physical degrees of freedom. Finally, we calculate the increase in the bond dimension after application of a finite-depth circuit and these isometry-maintaining transformations.

A unitary operator $U$ is $k$-local if it may be written as $U = U_1 \otimes U_2 \dots $ where $\supp(U_i)$ consists of $k$ geometrically local physical degrees of freedom and $\supp(U_i) \cap \supp(U_j) = \emptyset$ if $i\neq j$. A quantum circuit $U_c$ is $k$-local with depth $d$ if it may be written as a product $U_c = U^{(d)}U^{(d-1)}\dots U^{(1)}$ where each $U^{(i)}$ is $k$-local and does not necessarily commute with the other $U^{(i)}$. 

To apply a $k$-local unitary and maintain the isometric form, we first perform a ``coarse-graining" transformation, apply the unitary, and then perform a ``fine-graining" transformation. To illustrate these transformations, we use as an example a 4-local, depth-1 circuit in which each local unitary $U$ is supported on four physical degrees of freedom organized as shown in Fig.~\ref{fig:2D_circuit}. We depict these transformations for one of the local unitaries graphically in Fig.~\ref{fig:graining_moves}, keeping in mind that this transformation occurs for each $U$ in the circuit. We first perform the coarse-graining transformation by contracting the ancilla indices connecting the four tensors associated to the support of $U$ (Fig.~\ref{fig:graining_moves} (a)). This produces a  single tensor with physical dimension $d^4$ and bond dimension $\chi^2$ (Fig.~\ref{fig:graining_moves} (b)). Applying the unitary on the physical indices of the coarse-grained tensor preserves its isometries. Finally, we perform a fine-graining transformation to split the coarse-grained tensor back into four tensors, each with physical index of dimension $d$. There are different ways to construct the new tensors, but a particularly simple choice is to re-group the physical indices of a site with other incoming indices, incurring a cost that depends on the range of the unitary (Fig.~\ref{fig:graining_moves}(c)). As illustrated, the resulting tensors do not have the same bond dimension.

\begin{figure}
    \centering
    \includegraphics[width=0.8\linewidth]{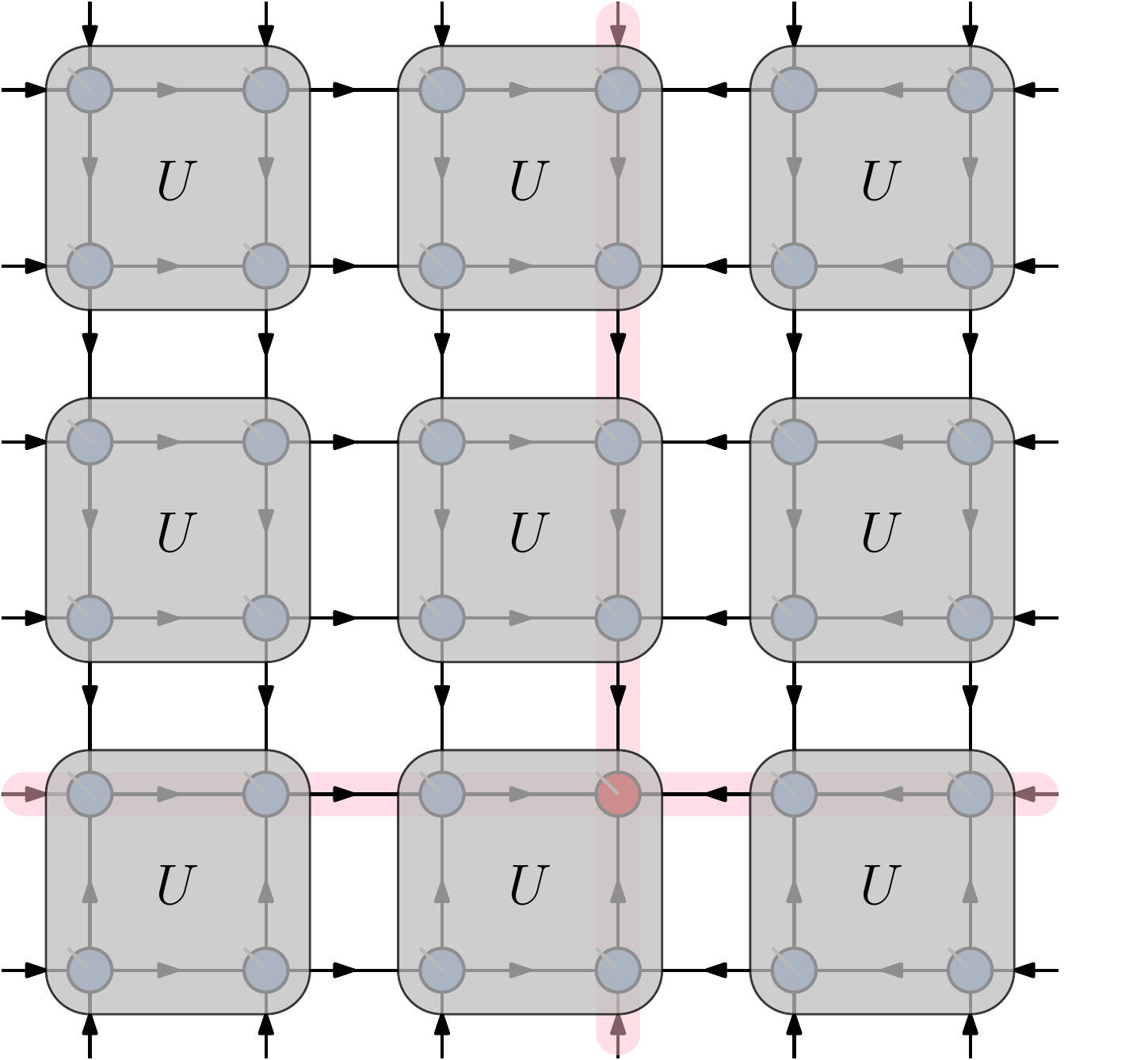}
    \caption{\label{fig:2D_circuit}Example of a local unitary quantum circuit of depth 1 acting on the physical Hilbert space of a 2D isoTNS.}
\end{figure}

We now consider the increase in the bond dimension due to a circuit of depth $D$. For example, in an MPS with physical Hilbert space dimension $d$ and bond dimension $\chi$, each unitary of a depth-1 $k$-local quantum circuit increases the bond dimension $\chi \rightarrow \chi d^{k-1}$. Continuing with the 2D example in Figs.~\ref{fig:2D_circuit} and \ref{fig:graining_moves}, each layer increases the bond dimension $\chi\rightarrow \chi^2 d^2$. Although the bond dimension is not uniform across the circuit, we consider the worst case increase. If subsequent layers of a depth-$D$ circuit in Fig.~\ref{fig:2D_circuit} have similar geometry, then the bond dimension increases as $\chi\rightarrow (\chi^{2}d^{2})^D$, independent of system size. For different $k$-local circuit geometries, the exponential factors will generally differ but the bond dimension will increase by $O(\exp(kD))$. The preceding analysis and conclusion therefore proceed similarly.

\begin{figure}
    \centering
    \subfigure[]{
        \includegraphics[scale=0.55]{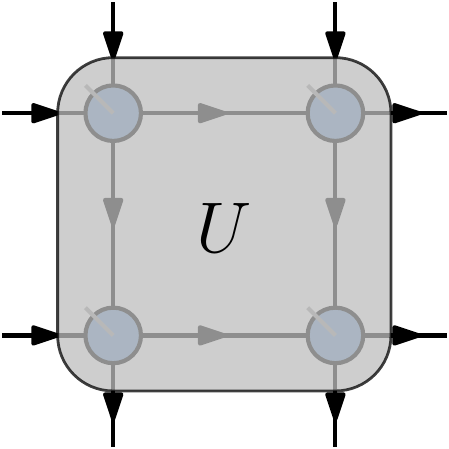}
        }
    \quad
    \subfigure[]{
        \includegraphics[scale=0.6]{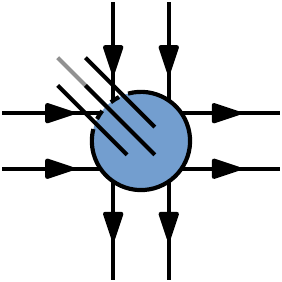}
    }

    \subfigure[]{
        \includegraphics[scale=0.9]{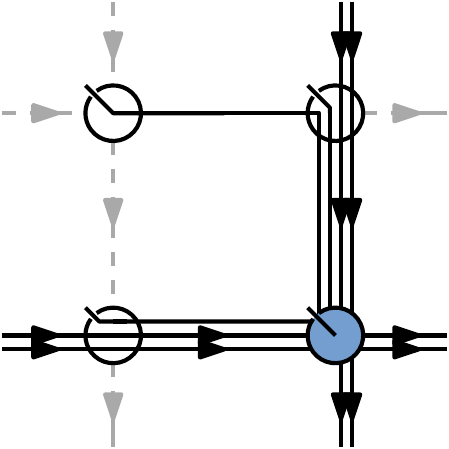}
    }
    \caption{\label{fig:graining_moves}Coarse-graining and fine-graining moves. (a) A local unitary from Fig.~\ref{fig:2D_circuit} acting in the physical Hilbert space. (b) Coarse-graining transformation. The four internal ancilla legs and the unitary is now effectively a one-site operation, where the tensor now has physical dimension $d^4$. (c) Fine-graining transformation. After the unitary is applied, the coarse-grained tensor is split and the original locations of the physical indices is restored. Dashed gray legs indicate a trivial ancilla with bond dimension 1. The top-left, top-right, and bottom-left tensors, denoted by empty circles, consist only of the Kronecker deltas indicated. The largest cost is incurred by the right leg, $\chi\rightarrow \chi^2 d^2$.}
\end{figure}

A subtle point is that it appears that the orthogonality hypersurface may not be exactly movable after applying the circuit. However, because it may be moved exactly within the fixed-point wavefunction, and the transformation described above preserves the isometries, the location of the orthogonality center will be inherited by the final state. Therefore, we conclude that all states which may be transformed into fixed-point string-net liquids by local quantum circuits of finite depth have exact 2D isoTNS representations of finite bond dimension.

A common wisdom is that approximating quasiadiabatic evolution by a finite depth circuit produces good approximations of local observables. This wisdom is motivated in a precise way in App.~\ref{app:short_depth_circuit}. In Sec.~\ref{sec:numerics}, for example, a family of perturbations to the fixed-point toric code tensors was considered. It was found to be well-approximated by an isoTNS of finite bond dimension. The context of this section considers a circuit which approximates the transformation of the ground state of the exact toric code to the perturbed states, which are in the same phase. The result of the analysis motivates why one might expect that isoTNS representations exist.

The above bound on the penalty of the bond dimension is not tight as we have demanded that at each step, an exact isoTNS representation is retained. In practice, however, this may not be necessary. However, the analysis of the relevant approximations required are beyond the scope of this paper.

\section{Conclusion}
We first showed numerically that a family of states in the same phase as the toric code can be put in an isometric form with constant error density. We then showed analytically that the PEPS description of string-net wavefunctions can be exactly put in the form of an isometric tensor network. Combined with the fact that even after applying finite depth circuits, isometries may be preserved (with a constant increase in bond dimension), these results show that isometric tensors are capable of representing a variety of interesting topological phases, including all abelian topological phases with a gappable edge. There are certain generalizations of string-net liquids \cite{Lan14} that are not covered by our proof, and it is possible that these, too, have exact isoTNS representations. We leave the problem of finding their isoTNS representations  to future work.

We make the following comments. Our construction of isometric string-net tensors relied on defining full-rank tensors by stripping off fusion constraints from the non-isometric tensors, and the two tensor networks are not related by gauge transformations alone. We speculate that generally, to obtain an equivalent isometric tensor network from a non-isometric counterpart, we may need to exploit redundancies besides gauge transformations. It remains to be seen what the consequences are for numerical algorithms that try to find isoTNS representations.

Finally, while we proved isoTNS descriptions exist, it remains to be proven that a local algorithm starting from either a non-isometrized TNS or Hamiltonian can find them. Furthermore, while we focused on finding exact isoTNS representations, it is of greater practical relevance  to understand how states may be \textit{approximated} by the isoTNS ansatz. In particular, we may be interested in how rapidly local expectation values or the fidelity converge with the bond dimension. More work on the algorithms side is required in order to address these issues, but we hope that the analytical results presented here provide motivation to initiate the required numerical work.

\begin{acknowledgments}
We thank Parsa Bonderson for clarifying technical details regarding fusion categories. We also thank Michael Levin for helpful correspondence regarding non-abelian string-net liquids. KS acknowledges support from the NSF Graduate Research Fellowship Program (Grant No. DGE 1752814). SC acknowledges support from the ERC synergy grant UQUAM. MZ and NB were supported by the DOE, office of Basic Energy Sciences under contract no. DE-AC02-05-CH11231. FP acknowledges the support of the Deutsche Forschungsgemeinschaft (DFG) Research Unit FOR 1807 through grants no. PO 1370/2-1, TRR80, the DFG under Germany's Excellence Strategy--EXC-2111-390814868 and the European Research Council (ERC) under the European Unions Horizon 2020 research and innovation program (grant agreement no. 771537).

\end{acknowledgments}

\appendix

\section{\label{app:isoTNS}Isometric Tensor Network Ansatz}

In this Appendix, we motivate and define the isometric tensor network state (isoTNS) ansatz presented in Ref.~\cite{Zaletel19} as well as clarify the graphical notation employed throughout this paper. The isoTNS ansatz generalizes the well-known isometry conditions that have made the matrix product state (MPS) ansatz an effective numerical tool for studying 1D wavefunctions. See Ref.~\cite{Schollwock11} for further details on MPS.

The MPS ansatz describes a wavefunction defined on a 1D chain of $N$ sites with open boundary conditions, each with an associated Hilbert space of dimension $d$. Each site $a$ ($2\leq a \leq N-1$) has a rank-3 tensor $T^{i_a}_{\alpha \beta}$ where the $i_a$ is referred to as the ``physical" index ($0 \leq i_a \leq d-1$) and $\alpha,\beta$ are ``ancilla" indices $1 \leq \alpha,\beta \leq \chi$ where $\chi$ is referred to as the ``bond dimension". Sites $a=1$ and $a=N$ have similarly defined rank-2 tensors. The wavefunction is defined by 

\begin{equation}
    \ket{\Psi} = \sum_{\vec{i}} T^{i_1}_{ \beta_1} T^{i_2}_{\beta_1 \beta_2}\dots T^{i_N}_{\beta_{N-1}} \ket{i_1 i_2 \dots i_N}
    \label{eq:MPS_ansatz}
\end{equation}
where $\vec{i} = (i_1, i_2, \dots i_N)$, the sum is taken over all vectors in $\mathbb{Z}_d^{\otimes N}$, and repeated indices are summed over (contracted). Graphically, tensors are denoted as solid shapes, such as circles and triangles as used in the main text. For each index a line emerges as in Fig.~\ref{fig:MPS_tensor}. Contraction of indices of two tensors is indicated graphically by joining the lines of the tensors corresponding to those indices.

\begin{figure}
\includegraphics[scale=0.55]{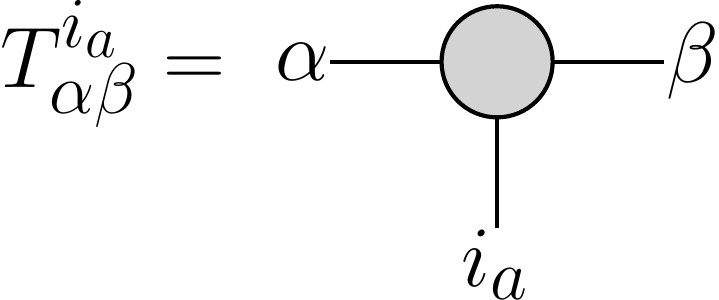}
\caption{\label{fig:MPS_tensor} Graphical notation of a tensor. }
\end{figure}
This ansatz contains a gauge degree of freedom that allows one to redefine the matrices $T$ without changing the wavefunction. For example, one can insert any invertible map $V$ between two adjacent tensors without modifying the wavefunction.

\begin{equation}
V^{-1}_{\alpha \beta} V_{\alpha' \beta} = \delta_{\alpha, \alpha'}
\end{equation}
Tensors  $\{T^{i_a}_{\alpha \beta}\}$ may be redefined as $T^{i_a}_{\alpha\beta}\rightarrow T^{i_a}_{\alpha\gamma}V^{-1}_{\beta\gamma}$ and $T^{i_{a+1}}_{\alpha\beta}\rightarrow V_{\alpha \gamma} T^{i_{a+1}}_{\gamma \beta}$. A tensor is left-isometric if $T^{i_a*}_{\alpha\beta}T^{i_a}_{\alpha\beta'} = \delta_{\beta \beta'}$ and right-isometric if $T^{i_a}_{\alpha\beta}T^{i_a*}_{\alpha'\beta} = \delta_{\alpha\alpha'}$. Graphically, a tensor satisfying isometry conditions is denoted by placing incoming arrows on indices which are contracted to obtain the identity in the remaining indices, denoted by placing outgoing arrows as in Fig.~\ref{fig:isometry_conditions}.

\begin{figure}[htp]
    \centering
    \subfigure[]{
        \includegraphics[scale=0.5]{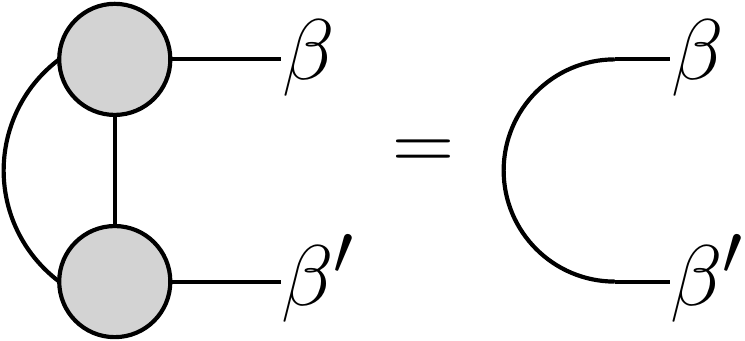}
        }
    \quad
    \subfigure[]{
        \includegraphics[scale=0.5]{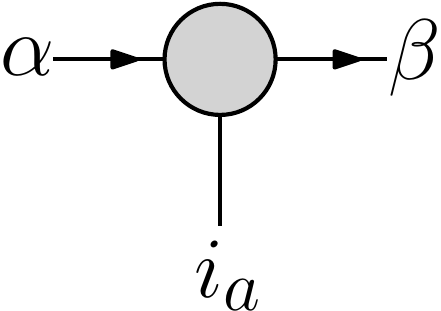}
    }
    \qquad
    \subfigure[]{
        \includegraphics[scale=0.5]{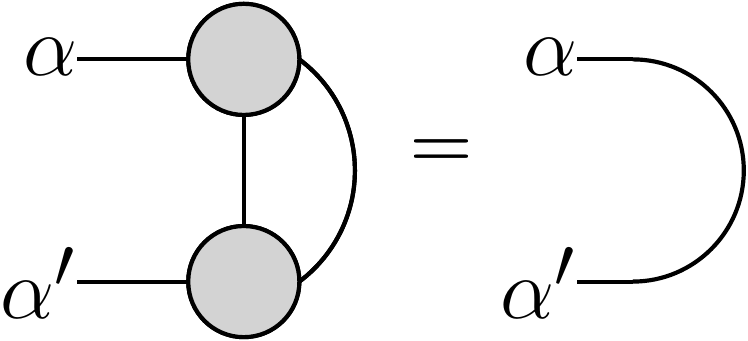}
    }
    \quad
    \subfigure[]{
        \includegraphics[scale=0.5]{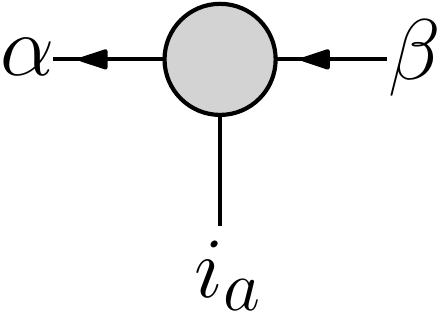}
    }        
    \caption{\label{fig:isometry_conditions}Graphical notation of isometries. (a) Left-isometric condition for an MPS tensor. (b) Arrow notation for an MPS tensor indicating that it satisfies condition (a). Physical indices always have an implicit incoming arrow. (c) Right-isometric condition for an MPS tensor. (d) Arrow notation for an MPS tensor indicating that it satisfies condition (c).}
\end{figure}

By suitable gauge transformations, for any site $a'$, tensors for $a < a'$ ($a > a'$) can be placed in left(right)-isometric form. In that case, we refer to site $a'$ as the orthogonality center. This representation allows for efficient computation of expectation values of operators supported compactly around site $a'$, as tensors outside the support need not be contracted explicitly.

Higher dimensional isoTNS ansatzes lift these isometry conditions to tensors networks of higher dimensions. For example, in the 2D generalization depicted in Fig.~\ref{fig:2D_isoTNS}, the 0-dimensional orthogonality center of the 1D MPS is now a 1D orthogonality hypersurface, which itself is an MPS with a 0D orthogonality center. The direction of the isometries implies that expectation values of operators supported on the orthogonality hypersurface can be computed without explicitly contracting tensors supported off the orthogonality hypersurface. If the orthogonality hypersurface can be moved as in the 1D case, then such a wavefunction may be numerically tractable. In this paper, we have provided a partial answer to the question of what physical states are described exactly by 2D tensor networks of this form.

\section{\label{app:string_net_liquid}String-Net Liquid}
In this Appendix, we review the construction of the string-net liquids and relate them to the PEPS tensor used in the main text. We restrict ourselves to isotopy-invariant and mirror-symmetric case considered in Ref.~\cite{Levin05}, although some generalizations exist. In that paper, string-net states were introduced to characterize various topological orders in 2+1D models. Many-body states are defined on a trivalent lattice with degrees of freedom on the edges. Each edge has an oriented string on it. The global string-net wavefunction is a weighted superposition of configurations of strings
\begin{equation}
    \ket{\Psi} = \sum_{X} \Phi(X) \ket{X}
\end{equation}
where $X$ refer to different configurations of strings and $\Phi(X)$ is a function from configurations to complex numbers. The following input information is necessary to define the function $\Phi(X)$ by defining the relationships between different configurations of strings.

\begin{enumerate}
    \item \textbf{String Types}: Each edge carries one of $N+1$ possible strings numbered $s=1\dots N$ or no string at all, numbered $s=0$. Strings are generally oriented so that $i^*$ denotes a string oriented opposite $i$. A string is unoriented if $i=i^*$.
    \item \textbf{Quantum Dimension}: Each string type $s$ has a value $d_s > 0$ associated to it, termed the ``quantum dimension''. The total quantum dimension $\mathcal{D}$ is defined as $\mathcal{D} = \sum_s d_s^2$.
    \item \textbf{Branching}: Not all configurations of strings are valid. One constraint on configurations that may appear are determined by which strings may meet at a vertex. These are referred to as ``branching rules" or ``fusion constraints", indicated by a symbol $\delta_{ijk}$ which is 0 if strings $ijk$ are forbidden from meeting at a vertex or 1 if permitted. By definition, fusion constraints are invariant under permutation of indices or switching the orientation of three strings at once. We assume they satisfy the following associativity constraint:
    \begin{equation}
        \sum_e \delta_{abe^*} \delta_{ecd} = \sum_f \delta_{afd^*}\delta_{bcf^*}
    \end{equation}
    \item $F$ \textbf{symbols}: The $F$-symbol is a six-index tensor describing how to recouple a string-net configuration. It defines how coefficients in the many body wavefunction $\Phi(X)$ differ between two configurations $X$ and $X'$ which differ by a local recoupling of strings. In Eq.~\ref{eq:string_net_f_move} below, we mean that $\Phi(X)$ for configuration $X$ on the LHS is related to $\Phi(X')$ for a set of configurations $X'$ parameterized by index $n$ on the RHS.
        \begin{equation}
            \includegraphics[scale=0.5]{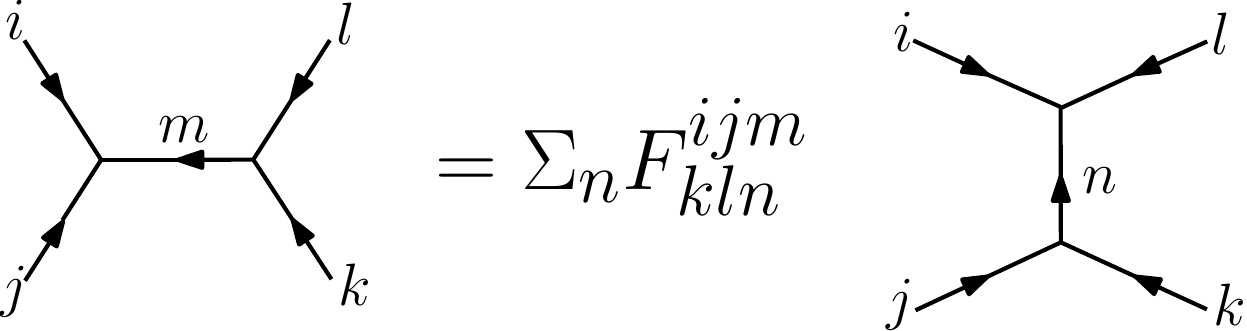}
            \label{eq:string_net_f_move}
        \end{equation}
    The quantum dimension is also encoded in the $F$-symbol as $d_s = 1/|F^{ss^*0}_{ss^*0} |$. We define the quantity $\ell_s = 1/F^{ss^*0}_{ss^*0}$ which may be negative and $v_s = \sqrt{\ell_s}$ which may be complex. By definition, we have
    \begin{equation}
        \sum_c \delta_{abc^*} \ell_c = \ell_a \ell_b
    \end{equation}
\end{enumerate}

We now summarize properties of $F$-symbols used throughout the main text. The $F$-symbols obey a relation referred to as the pentagon equation
\begin{equation}
    \sum_{n=0}^N F^{mlq}_{kp^*n}F^{jip}_{mns^*}F^{js^*n}_{lkr} = F^{jip}_{q^*kr^*}F^{riq^*}_{mls^*}
    \label{eq:string_net_pentagon}
\end{equation}
and may be chosen to satisfy the normalization
\begin{equation}
    F^{ijk}_{j^*i^* 0} = \frac{v_k}{v_iv_j} \delta_{ijk}
    \label{eq:string_net_normalization}
\end{equation}
They can be shown to obey a tetrahedral symmetry property:
\begin{equation}
    F^{ijm}_{kln} = F^{jim}_{lkn^*} = F^{lkm^*}_{jin} = \frac{v_m v_n}{v_j v_l} F^{imj}_{k^*nl}
    \label{eq:string_net_tetrahedral}
\end{equation}
The $F$-symbols must satisfy the following ``unitarity" condition.

\begin{equation}
    F^{i^*j^*m^*}_{k^*l^*n^*} = (F^{ijm}_{kln})^*
    \label{eq:string_net_unitarity}
\end{equation}
The fact this is a unitarity condition can be seen from the pentagon equation and by using the properties listed above, where we set $j = k^*$, $l = s^*$ and $r = 0$. Then,

\begin{eqnarray}
    \nonumber \sum_{n=0}^N F^{ms^*q}_{j^*p^*n}F^{jip}_{mns^*}F^{js^*n}_{sj^*0} &=& F^{jip}_{q^*j^*0}F^{0iq^*}_{ms^*s^*}\\
    \nonumber \sum_{n=0}^N F^{j^*p^*q^*}_{msn^*}F^{jpi}_{m^*s^*n} &=& \delta_{p,i}\delta_{jip}\delta_{sqm}\\
    \sum_{n=0}^N (F^{jpq}_{m^*s^*n})^*F^{jpi}_{m^*s^*n} &=& \delta_{p,i}\delta_{jip}\delta_{sqm}
\end{eqnarray}

\noindent where we used tetrahedral symmetry and normalization to go to the second line and Eq.~\ref{eq:string_net_unitarity} to go to the third line. Therefore the ``unitarity condition'' indeed imply that $F^{ijm}_{kln}$ is a unitary matrix when viewed as a matrix with indices $m,n$.

We also note the following equality holds as a result of Eqs.~\ref{eq:string_net_normalization} and~\ref{eq:string_net_pentagon}.

\begin{equation}
    \sum_c \delta_{abc^*} d_c = d_a d_b
\end{equation}

\begin{figure}
    \centering
    \includegraphics[scale=0.7]{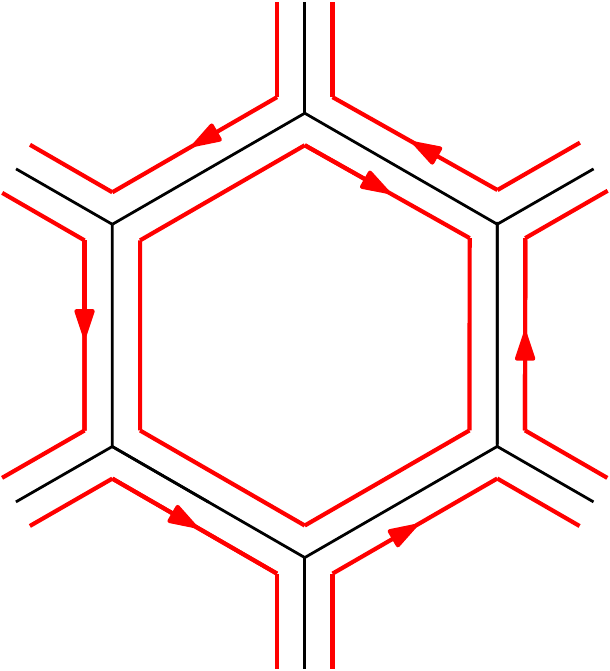}
    \caption{\label{fig:thickened_lattice}Thickened honeycomb lattice. Physical degrees of freedom are situated on the black edges. Each plaquette carries an internal loop carrying a string.}
\end{figure}

To understand these states as well as construct their tensors, we begin with a thickened honeycomb lattice in which in addition to the original edges, every hexagonal face contains an additional internal loop as in Fig.~\ref{fig:thickened_lattice} as in Ref.~\cite{Gu2009}. Although the final wavefunction is defined on the original honeycomb lattice, we can construct it by first creating a superposition of internal loops. The internal loops are then coupled to the degrees of freedom on the original honeycomb lattice using $F$-moves. First, define the operator $B_p$ associated to plaquette $p$:
\begin{equation}
    B_p = \sum_{s=0}^N \frac{\ell_s}{\mathcal{D}} B_p^s
\end{equation}
where $B_p^s$ creates a string of type $s$ on plaquette $p$. The operator $B_p$ then creates a superposition of strings on the internal loop at plaquette $p$. The wavefunction on the internal loops is then defined by starting with $\ket{0}$, the state where all internal loops are the 0-string:

\begin{equation}
    \ket{\Psi_{\text{loops}}} = \prod_p B_p\ket{0} \sim \sum_{\vec{s}}\left(\prod_{i} \ell_{s_i}\right) \ket{\vec{s}}
\end{equation}
The wavefunction defined on the edges of the original honeycomb lattice is obtained by fusing the loops on both sides of each edge according to Eq.~\ref{eq:string_net_f_move} as in Fig.~\ref{fig:string_net_recoupling}. We use this prescription to construct a tensor for each vertex, depicted as triangles in Fig.~\ref{fig:honeycomb_lattice} (b). The resulting tensor in the bulk takes the form

\begin{equation}
    T^{i \alpha \beta}_{\gamma j k} = \sqrt{\frac{v_i v_j}{v_k}} F^{i\alpha^* \beta}_{\gamma j k} \delta_{\alpha \gamma^* k}\delta_{\beta^*\gamma j}\delta_{\alpha^* \beta i}
\end{equation}
The tensor on the other sublattice is obtained similarly, and they are related by complex conjugation:

\begin{equation}
    T^{l \alpha \sigma}_{\gamma m k}  =  \sqrt{\frac{v_l v_m}{v_k}} F^{l^* \alpha \sigma^*}_{\gamma^* m^* k^*} \delta_{\alpha \gamma^* k}\delta_{\sigma^*\gamma m}\delta_{\alpha^* \sigma l}
\end{equation}

\begin{figure}
    \centering
    \subfigure[]{
        \includegraphics[scale=0.40]{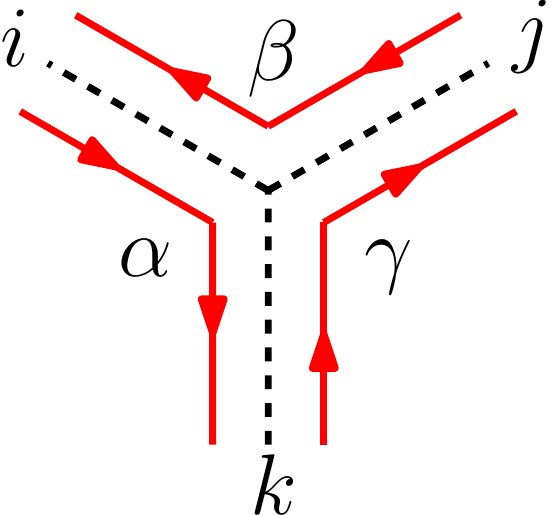}
        }
    \quad
    \subfigure[]{
        \includegraphics[scale=0.40]{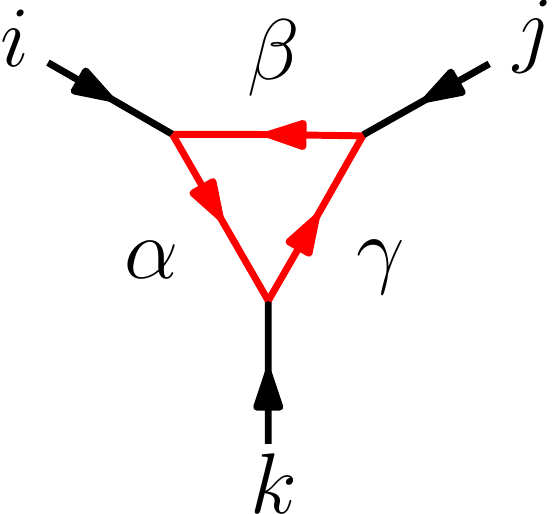}
    }
    \quad
    \subfigure[]{
        \includegraphics[scale=0.40]{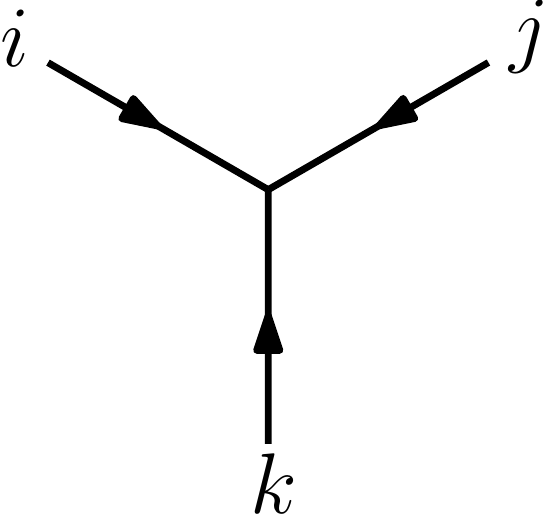}
    }
    \caption{\label{fig:string_net_recoupling}Schematic for obtaining string-net wavefunction via $F$-moves. (a) $B_p^s$ on each plaquette surrounding the vertex has placed a loop of types $\alpha,\beta,\gamma$. (b) $F$-move is applied to couple internal loops on the plaquettes to strings on the edges. (c) The remaining bubble is eliminated.}
\end{figure}
Tensors along the edge of the system (the outermost tensors in Fig.~\ref{fig:honeycomb_lattice} (b)) may be similarly derived. The edge tensors amount to adding different factors of $v_i$ depending on where on the edge it is (e.g. corners and edges carry different factors of $v_i$).

\section{\label{app:f_symbols_identities}Relations of \texorpdfstring{$F$}{F}-symbols}
The input data used to define a string-net model defined in App.~\ref{app:string_net_liquid} is closely related to the algebraic data specifying a unitary fusion category (UFC). In the main text, we have studied the string-net models as first proposed in Ref.~\cite{Levin05}, which obey additional properties beyond those of the most general UFC. The language and rules of the original string-nets of Levin and Wen and those of general UFCs are therefore slightly different, and in this Appendix, we sketch how properties of the $F$-symbols used to define Levin-Wen string-net models correspond to properties of the $F$-symbols that define their corresponding UFCs. We emphasize again that these properties may not hold for the most general UFCs. Levin and Wen's string-nets obey strict rotational and isotopy invariance and correspond to some strict subset of all UFCs. This appendix assumes familiarity with notations and definitions in UFC literature, so we refer the reader to Ref.~\cite{Bonderson08} for an introduction to the UFC framework.

For every Levin-Wen $F$-symbol, we define a corresponding UFC $F$-symbol as follows:
\begin{equation}
    [F^{abc}_{d}]_{ef} = F^{b^*a^*e}_{dc^*f}
\end{equation}
These are the ordinary, unbent $F$-symbols of the UFC. We discuss bent versions below. Next, branching rules correspond to fusion constraints as
\begin{equation}
        N_{ab}^{c^*} = N_{ac}^{b^*} = N_{bc}^{a^*} = \delta_{abc}
\end{equation}

These identifications allow us to immediately find several identities for the UFC $F$-symbols $[F^{abc}_d]_{ef}$ using properties from App.~\ref{app:string_net_liquid}. First, we recover the well-known UFC version of the pentagon equation:

\begin{equation}
    [F^{pcd}_{e}]_{qr} [F^{abr}_{e}]_{ps} = \sum_{x} [F^{abc}_{q}]_{px} [F^{axd}_{e}]_{qs} [F^{bcd}_{s}]_{xr}
    \label{eq:ufc_pentagon}
\end{equation}
Next, unitarity of the ordinary (unbent) $F$-symbols within the subspaces allowed by fusion constraints may be seen from the orthonormality of the string-net recoupling moves:

\begin{eqnarray}
    \sum_f \Fst{a}{b}{c}{d}{e}{f}\Fst{a}{b}{c}{d}{e'}{f}^* &=& \sum_f F^{b^*a^*e}_{dc^*f}(F^{b^*a^*e'}_{dc^*f})^* \\
    &=& \delta_{e,e'} \delta_{b^*a^*e} \delta_{ced^*} \\
    &=& \delta_{e,e'} N_{ab}^e N_{ce}^d
    \label{eq:ufc_unitarity}
\end{eqnarray}
The next two properties may not hold for general UFCs. First, the tetrahedral symmetries from Eq.~\ref{eq:string_net_tetrahedral} are translated as follows:

\begin{equation}
    [F^{abc}_{d}]_{ef} = [F^{bad^*}_{c^*}]_{ef^*} = [F^{d^*cb}_{a^*}]_{e^*f}
    \label{eq:ufc_tetrahedral}
\end{equation}
The second equality in Eq.~\ref{eq:string_net_tetrahedral} immediately yields a mirror symmetry:
\begin{equation}
    [F^{\alpha ij}_{\gamma}]_{\beta k} = [F^{\gamma^* ji}_{\alpha^*}]_{\beta^* k}
\end{equation}
Lastly, we can show that bent $F$-symbols satisfy a unitarity condition similar to Eq.~\ref{eq:ufc_unitarity}:
\begin{eqnarray}
    \nonumber \sum_e \Fbent{a}{b}{c}{d}{e}{f}\Fbent{a}{b}{c}{d}{e}{f'}^* &=& \sum_e \frac{d_e \sqrt{d_f d_{f'}}}{d_ad_d} \Fbtos{a}{b}{c}{d}{e}{f}^* \Fbtos{a}{b}{c}{d}{e}{f}\\
    \nonumber &=& \sum_e \frac{d_e \sqrt{d_f d_{f'}}}{d_ad_d} (F^{e^*c^*a}_{fb^*d})^*F^{e^*c^*a}_{f'b^*d}\\
    \nonumber&=& \sum_e (F^{b^* a^* f}_{c d e^*})^*F^{b^* a^* f'}_{c d e^*}\\
    \nonumber&=& \delta_{f,f'} \delta_{b^* a^* f} \delta_{c^*d^*f}\\
    &=& \delta_{f,f'} N_{ab}^{f} N_{cd}^{f}
\end{eqnarray}
where it is important to recall that $v_i = \sqrt{\ell_i}$ may be complex. Non-trivial Frobenius-Schur indicators do not cause a problem, as $v_i \times v_i^* = d_i$ \cite{Bonderson08}.

\section{\label{app:full_rank_tensors}Constructing the \texorpdfstring{$A$}{A}-symbols}
In this appendix,  we show in detail how to construct the $A$-symbols $A^{ijm}_{kln}$ defined to satisfy the properties in Eq.~\ref{eq:full_rank_tensors} in the main text, which we repeat here for convenience:

\begin{equation}
    \begin{aligned}
        F^{ijm}_{kln} = A^{ijm}_{kln} \delta_{ijm} \delta_{klm^*}&\\
        \sum_{n=0}^N (A^{ijm}_{kln})^*A^{ijm'}_{kln} = \delta_{m,m'}&
    \end{aligned}
\end{equation}

We first fix $i,j,k,l$. Due to fusion constraints, there is a limited choices for $m,n$ that gives nonzero $F^{ijm}_{kln}$. We can count the number of those indices as follows.

\begin{equation}
    \delta_{ijkl} \stackrel{\text{def}}{=} \sum_m \delta_{ijm} \delta_{klm^*} = \sum_n \delta_{iln} \delta_{jkn^
*}
\end{equation}

\noindent where the second equality is guaranteed by the associativity of the fusion constraints. We permute indices for $m$ and $n$ such that the first $\delta_{ijkl}$ indices satisfy the fusion constraints. Then, seen as a matrix in $m$ and $n$, $F^{ijm}_{kln}$ has the following structure as a result of Eq.~\ref{eq:fusion_constraints} and Eq.~\ref{eq:F_unitary_in_subspace}.

\begin{equation}
    F^{ijm}_{kln} =
    \begin{bmatrix}
    F^{ij}_{kl} & 0\\
    0 & 0\\
    \end{bmatrix}
\end{equation}

\noindent where $F^{ij}_{kl}$ is a $\delta_{ijkl} \times \delta_{ijkl}$ unitary matrix that corresponds to the nonzero part of $F^{ijm}_{kln}$. In order to make this into a full rank unitary matrix, we populate the bottom right by an identity matrix as in the following:

\begin{equation}
    A^{ijm}_{kln} =
    \begin{bmatrix}
    F^{ij}_{kl} & 0\\
    0 & I\\
    \end{bmatrix}
\end{equation}

This matrix is unitary in the full, unconstrained space. By construction, we have $A^{ijm}_{kln} \delta_{ijm} \delta_{klm^*} = A^{ijm}_{kln}\delta_{iln} \delta_{jkn^
*} = F^{ijm}_{kln}$, so this $A$-symbol satisfies the properties of Eq.~\ref{eq:full_rank_tensors}.

\subsection{\label{app:1_string_model}\texorpdfstring{$N=1$}{N=1} string model}
We will look at the simplest string-net model, which only has a vacuum string ($0$) and one non-trivial string ($1$). The corresponding fusion constraints are given by

\begin{equation}
    \delta_{000} = \delta_{110} = 1
\end{equation}
Other combinations of indices that are not permutations of the above two are all zero. The $F$-symbols are trivial and given by the product of fusion constraints:

\begin{equation}
    F^{ijm}_{kln} = \delta_{ijm} \delta_{klm^*} \delta_{iln} \delta_{jkn^*}
\end{equation}
The resulting string-net model is known to produce the toric code Hamiltonian \cite{Levin05}.

To construct the full-rank tensors for the toric code, we fix the values of $i,j,k,l$ and look at the resulting matrices. For concreteness, we will look at two choices of $i,j,k,l$:

\begin{equation}
    F^{00m}_{00n} = 
    \begin{bmatrix}
    1 & 0\\
    0 & 0\\
    \end{bmatrix},\hspace{0.5 cm}
    F^{01m}_{00n} = 
    \begin{bmatrix}
    0 & 0\\
    0 & 0\\
    \end{bmatrix}\, ,
\end{equation}

\noindent where matrices on the right hand side are labeled by string-types in the order  $0,1$. The corresponding entries of the full rank tensor $A^{ijm}_{kln}$ are then given by

\begin{equation}
    A^{00m}_{00n} = 
    \begin{bmatrix}
    1 & 0\\
    0 & \textcolor{red}{1}\\
    \end{bmatrix},\hspace{0.5 cm}
    A^{01m}_{00n} = 
    \begin{bmatrix}
    \textcolor{red}{1} & 0\\
    0 & \textcolor{red}{1}\\
    \end{bmatrix}\, ,
\end{equation}
where we used red to indicate which entries were changed from the original $F$-symbol.

\subsection{\label{app:ising_anyon_model}The Ising anyon model}
As a more non-trivial example, we look at the Ising string-net model. The model has two non-trivial string types, $\frac{1}{2},1$, in addition to the vacuum string $0$. Non-zero fusion constraints are given by

\begin{equation}
    \delta_{000} = \delta_{110} = \delta_{\frac{1}{2}\frac{1}{2}0} = \delta_{\frac{1}{2}\frac{1}{2}1} = 1
\end{equation}

\noindent and their permutations. Some of the $F$-symbols are given by

\begin{equation}
    F^{\frac{1}{2}\frac{1}{2}m}_{\frac{1}{2}\frac{1}{2}n} = 
    \begin{bmatrix}
    \frac{1}{\sqrt{2}}& \frac{1}{\sqrt{2}} & 0 \\
    \frac{1}{\sqrt{2}} & -\frac{1}{\sqrt{2}}&0\\
    0 & 0 & 0\\
    \end{bmatrix},
    F^{\frac{1}{2}1m}_{\frac{1}{2}1n} = 
    \begin{bmatrix}
    0 & 0 & 0 \\
    0 & 0 & 0\\
    0 & 0 & -1\\
    \end{bmatrix}\,,
\end{equation}

\noindent where matrix indices on the right are ordered as $0, 1, \frac{1}{2}$. The corresponding entries of the full rank tensor are given by

\begin{equation}
    A^{\frac{1}{2}\frac{1}{2}m}_{\frac{1}{2}\frac{1}{2}n} = 
    \begin{bmatrix}
    \frac{1}{\sqrt{2}}& \frac{1}{\sqrt{2}} & 0 \\
    \frac{1}{\sqrt{2}} & -\frac{1}{\sqrt{2}}&0\\
    0 & 0 & \textcolor{red}{1}\\
    \end{bmatrix},
    A^{\frac{1}{2}1m}_{\frac{1}{2}1n} = 
    \begin{bmatrix}
    \textcolor{red}{1} & 0 & 0 \\
    0 & \textcolor{red}{1} & 0\\
    0 & 0 & -1\\
    \end{bmatrix}
\end{equation}

\section{\label{app:abelian_isometry} Isometric Tensor for Abelian String-Net}
In this section, we show that the isometric tensor construction can be generalized to abelian models that were not covered in  original string-net construction of Ref. \cite{Levin05}. In particular, this construction captures all models given in \cite{Lin14}.

Let us first review what we mean by an abelian string-net model. An abelian string-net model satisfies the following condition: given two string types $a, b$, there is only one string type $c$ such that $\delta_{abc} = 1$. As a consequence, $d_a = 1$ for all string types. We also require $|F^{i\alpha\beta}_{\gamma j k}| = 1$ when it satisfies fusion constraints, so as to ensure unitarity. 

Now, consider a tensor of the form

\begin{equation}
    \tilde{F}^{i\alpha\beta}_{\gamma jk} = \mathcal{N}^{i\alpha\beta}_{\gamma jk} F^{i\alpha\beta}_{\gamma jk}
\end{equation}

\noindent where $|\mathcal{N}| = 1$ is a phase and $F$ is the $F$-symbol for some abelian string-net model. The tensor network representation for Levin-Wen string-net \cite{Levin05} as well as its generalizations \cite{Lin14} can be put in this form. This defines a triangular tensor as found in Fig.~\ref{fig:tensor_definition}. As in the main text, we need to ``strip off'' the fusion constraints such that the resulting tensor becomes isometric. Since there are three isometry directions, we need to define three six-index objects, $A$, $B$, and $C$. These are generalizations of the $A$-symbols in the main text and satisfy the following conditions:

\begin{equation}
    \sum_k A^{i\alpha\beta}_{\gamma jk}A^{i\alpha\beta'}_{\gamma jk} = \delta_{\beta,\beta'}, F^{i\alpha\beta}_{\gamma jk} = A^{i\alpha\beta}_{\gamma jk}\delta_{i \alpha \beta} \delta_{\gamma j \beta^*} 
\end{equation}

\begin{equation}
    \sum_j B^{i\alpha\beta}_{\gamma jk}B^{i\alpha'\beta}_{\gamma jk} = \delta_{\alpha,\alpha'},
         F^{i\alpha\beta}_{\gamma jk} = B^{i\alpha\beta}_{\gamma jk}\delta_{i\alpha\beta}\delta_{\alpha\gamma k^*} 
\end{equation}

\begin{equation}
         \sum_i C^{i\alpha\beta}_{\gamma jk}C^{i\alpha\beta}_{\gamma' jk} = \delta_{\gamma,\gamma'},
         F^{i\alpha\beta}_{\gamma jk} = C^{i\alpha\beta}_{\gamma jk}\delta_{\gamma j \beta^*} \delta_{\alpha \gamma k^*} 
\end{equation}

\noindent where we ignored $\mathcal{N}$ since it is just a phase and does not affect unitarity. Note that we only needed the $A$-symbols in the main text, since we could change isometry directions using tetrahedral symmetry. Here, we do not assume the symmetry exists, and instead rely on the fact that we are dealing with abelian string-net.

For concreteness, we focus on constructing a full-rank tensor $C$. In the spirit of App.~\ref{app:full_rank_tensors}, we fix four indices $\alpha, j, \beta, k$ and look at the matrix in terms of $i, \gamma$. Since we are dealing with abelian string-net $F$-symbols, there is at most one choice of $i, \gamma$ that gives nonzero $F$-symbol. Therefore, after appropriate permutation of indices, the matrix takes one of the following two forms.

\begin{equation}
    F^{i\alpha\beta}_{\gamma jk} =
    \begin{bmatrix}
    e^{i\theta_{\alpha \beta j k}} & 0 \\
    0 & 0 \\
    \end{bmatrix}
    , F ^{i\alpha\beta}_{\gamma jk} = 
    \begin{bmatrix}
    0 & 0 \\
    0 & 0 \\
    \end{bmatrix}
\end{equation}

\noindent where $\theta_{\alpha\beta j k}$ is a real number that depends on four indices, and bottom-right corner of the matrix is a $N \times N$ square matrix. It is then straightforward to construct the full-rank tensor:

\begin{equation}
    C^{i\alpha\beta}_{\gamma jk} =
    \begin{bmatrix}
    e^{i\theta_{\alpha \beta j k}} & 0 \\
    0 & \textcolor{red}{I} \\
    \end{bmatrix}
    , C ^{i\alpha\beta}_{\gamma jk} = 
    \begin{bmatrix}
    \textcolor{red}{1} & 0 \\
    0 & \textcolor{red}{I} \\
    \end{bmatrix}
\end{equation}

\noindent The other two directions go similarly. Note that $N=1$ string model in App.~\ref{app:1_string_model} is a special case of this construction. This construction proves the triangular tensor can be put in an isometric form in any direction. The proof of the isometry for a unit cell immediately follows by using the argument in Fig.~\ref{fig:isometry_proof}.

\section{\label{app:entanglement_spectrum} Entanglement spectrum}
In this section, we calculate the entanglement spectrum of the string-net liquid by using the isometric form. Our starting point is Fig.~\ref{fig:isometric_and_orthogonal}. Here we have four regions with different isometry directions and a orthogonality hypersurface. Note that the orthogonality hypersurface has physical indices coming out of it. In order to calculate the entanglement spectrum, however, it is better to collapse the orthogonality hypersurface to a line such that it only has ancilla indices coming out of it. This can be done by contracting physical indices of the orthogonality hypersurface. The resulting ``slim'' orthogonality hypersurface is given in Fig.~\ref{fig:orthogonality_center_for_entanglement_spectrum} (a). This is a boundary between four regions. In order to calculate bipartite entanglement spectrum, we can trace out unnecessary boundaries by using the relationship in Fig.~\ref{fig:orthogonality_center_for_entanglement_spectrum} (b). The result is given in (c). This is already in a diagonal form, and contracting this with its complex conjugate gives the density matrix. Remembering to take the fusion constraints into account, the density matrix is given by

\begin{eqnarray}
    \nonumber \rho &=& \sum_{\{q\},\{s\}} \left(\prod_{i} d_{q_i}\right) \delta_{q_0q_1 s_1} \delta_{s_1 q_2 s_2} ... \delta_{s_{N-1} q_N q_{N+1}}\\
    &\times&\ket{\{q\},\{s\}}\bra{\{q\},\{s\}}
\end{eqnarray}

\noindent where $\{q\}$ labels physical legs and $\{s\}$ label ancilla legs. This result agrees with the one in \cite{Levin06}, and thus provides a non-trivial consistency check on our isometric tensor construction.

\begin{figure*}
    \centering
    \includegraphics[width = 0.8\linewidth]{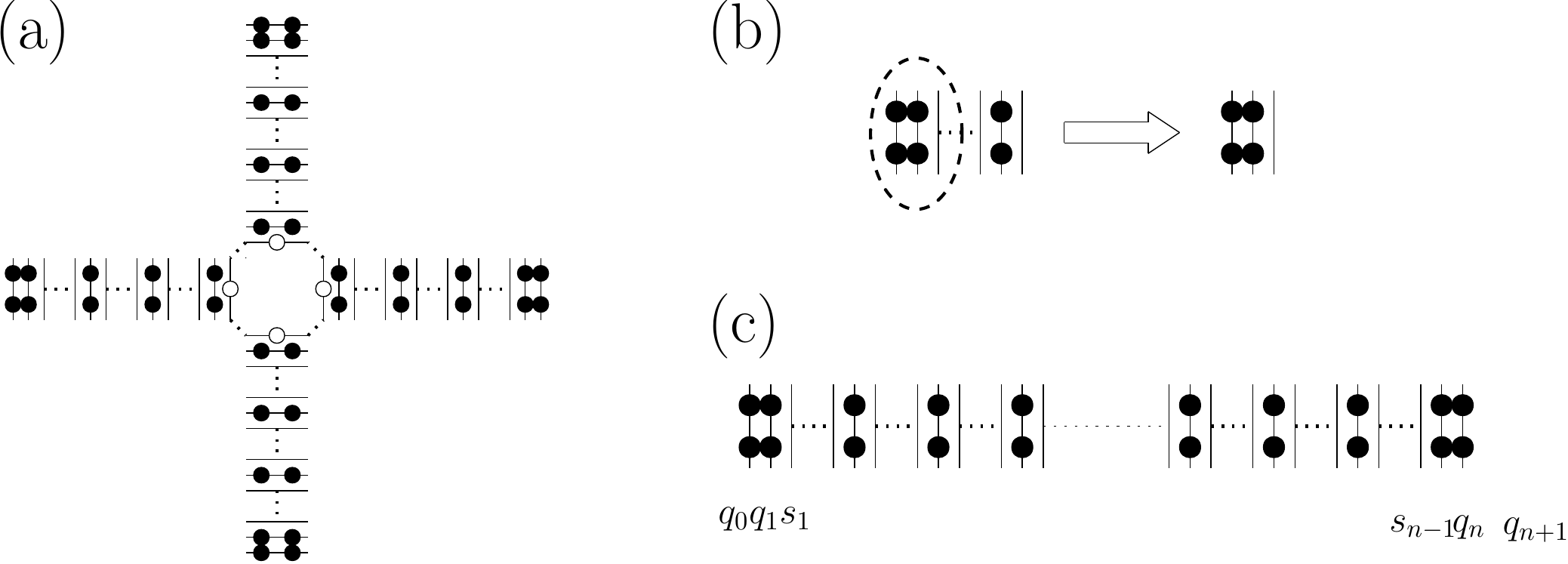}
    \caption{(a) Orthogonality hypersurface given by tracing over the physical indices of the orthogonality hypersurface in Fig.~\ref{fig:isometric_and_orthogonal}. Dotted lines mean those bonds have the same value. Each set of three bonds also satisfy fusion constraints, which are not explicitly shown in the figure. (b) Tracing over some of the degrees of freedom (circled by dashed line) brings quantum dimensions to the neighboring bond. (c) Boundary between the top region and the bottom region. Physical legs corresponding to physical degrees of freedom (labeled by $q$) carry factors of quantum dimensions while ancilla degrees of freedom (labeled by $s$) do not.}
    \label{fig:orthogonality_center_for_entanglement_spectrum}
\end{figure*}

\section{\label{app:short_depth_circuit}Finite-depth Circuit}

In this Appendix, we provide a brief analytical justification for considering wavefunctions which may be transformed into string-net wavefunctions by quantum circuits of finite depth. In particular, we make precise the common association of local adiabatic evolution with transformation by a circuit of constant depth.

First we set up some notations. We consider a many-body time-dependent Hamiltonian $H(t)$ defined on a set of degrees of freedom $\Lambda$, each with a Hilbert space $\mathbb{C}^d$, associated to the vertices of a lattice. $\Lambda$ is equipped with a distance. Distances between degrees of freedom $i$ and $j$ are denoted $\dist(i,j)$, and distances between subsets $X,Y\subset \Lambda$ are defined as $\dist(X,Y) = \min_{i \in X, j \in Y}\dist(i,j)$. The diameter of a set $X$ is defined as $\diam(X) = \max_{i,j\in X}\dist(i,j)$. Every set $X$ with $\diam(X) = 1$ has $|X| = O(1)$ (constant density of degrees of freedom). If an operator acts nontrivially on degrees of freedom in a set $X$, it is simply said to act on a set $X$. The support of an operator $O$ is the smallest set $Y\subset \Lambda$ such that for $\Lambda \tbs Y$, $O$ acts as the identity. We consider local time-dependent Hamiltonians $H(t)$, i.e. those which may be written as

\begin{equation}
    H(t) = \sum_X h_X(t)
\end{equation}
where $h_X(t)$ are terms supported on subsets $X \subset \Lambda$ such that $\diam(X) \leq 1$. This guarantees geometric locality of the Hamiltonian. Finally, as we concentrate on time-dependent Hamiltonians, $U_H(t)$ is defined by satisfying $i\partial_t U_H(t) = H(t)U_H(t)$:
\begin{equation}\label{eq:time_dep_unitary}
    U_H(t) = \mathcal{T}\exp\left(-i\int_0^t H(t')dt'\right)
\end{equation}
where $\mathcal{T}$ is the time-ordering operator. Below we concentrate on a local observable $O_Y$ supported on a set $Y$, and the $O(\cdot )$ notation used will carry its usual definition but will additionally suppress factors of $t$, as this is fixed by the quasiadiabatic evolution.

We first recall the standard definition of quantum phases:
\begin{definition}[Phases of gapped Hamiltonians]
Two local, gapped Hamiltonians $H_0$ and $H_1$ are in the same phase if there exists a smooth sequence of Hamiltonians $H(t)$ parameterized by a parameter $t$ such that $H(0) = H_0$,  $H(1) = H_1$, and $H(t)$ is local and gapped for all $t$. 
\end{definition}
Phases of quantum states which are the ground states of local, gapped Hamiltonians are defined similarly. In Ref.~\cite{Hastings05} it was shown that adiabatic evolution of a state may be locally approximated by a local, time-dependent Hamiltonian:

\begin{lemma}[Quasi-adiabatic Evolution~\cite{Hastings05}]\label{lemma:quasiadiabatic} If $\ket{\psi_0}$ and $\ket{\psi_1}$ are in the same phase, then for any local observable $O_Y$ and error $\epsilon$ there exists a smoothly varying local, time-dependent  Hamiltonian $H(t)$ such that:
\begin{equation}
    \mtude{\bra{\psi_1}O_Y\ket{\psi_1}-  \bra{\psi_0}U_H^\dagger (1) O_YU_H(1)\ket{\psi_0}} \leq  \norm{O_Y}\epsilon
\end{equation}
\end{lemma}
The locality of the Hamiltonian depends on the threshold $\epsilon$ and the set $Y$. We now translate this approximation into a circuit of finite-depth and state the main result of this Appendix.

\begin{theorem}\label{thm:adiabatic_evolution_fdc}
If $\ket{\psi_0}$ and $\ket{\psi_1}$ are in the same phase, then for any local observable $O_Y$ and error $\epsilon$ there exists an $O(1)$-depth local quantum circuit $U_c$ such that
\begin{equation}
    \mtude{\bra{\psi_1}O_Y\ket{\psi_1}-  \bra{\psi_0}U_c^\dagger O_Y U_c\ket{\psi_0}} \leq \norm{O_Y}  \epsilon
\end{equation}
and $O(\cdot)$ refers to system-size dependence and suppresses time dependence.
\end{theorem}
This justifies the use of finite-depth circuits to characterize the phase diagram of string-net liquids. To prove this, we work in the Heisenberg picture and prove two approximation lemmas. The first is an application of the Lieb-Robinson bound:
\begin{lemma}[Lieb-Robinson Bound~\cite{Lieb72, Nachtergaele06, Hastings06}]\label{lemma:lieb_robinson} Let $H(t)$ be a local Hamiltonian and $O_Y$ be an operator supported on a set $Y$. Then there exist positive constants $v,\mu$  such that for any $r \geq vt$ and set $X$ with $\dist(X,Y) = r$, 

\begin{equation}
    \norm{[O_Y(t), O_X]} \leq O(\norm{O_Y}\norm{O_X}\mtude{Y}e^{-\mu r} )
\end{equation}
where $O(\cdot)$ notation suppresses dependence on numerical constants and $t$, which are not relevant to this analysis.
\end{lemma}

\begin{figure}
    \centering
    \includegraphics[width=.7\linewidth]{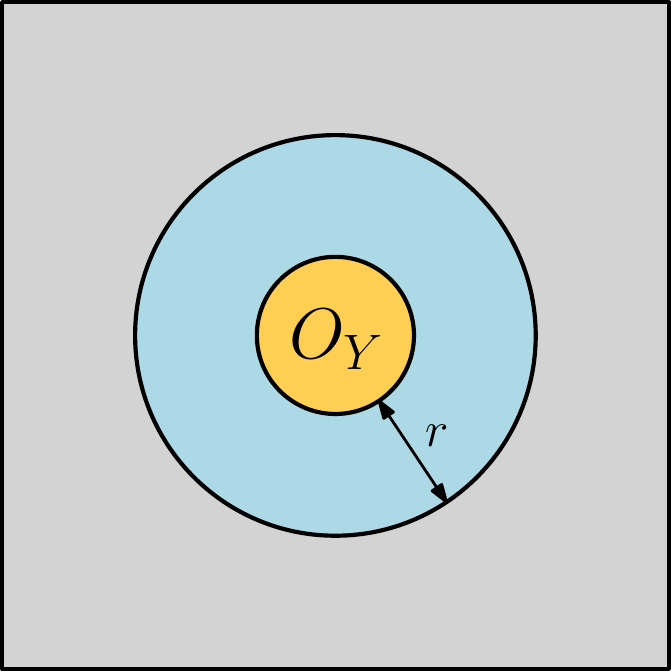}
    \caption{\label{fig:local_observable}Schematic representation of partition of $\Lambda$ used in Eq.~\ref{eq:partitioned_H}. Set $R$ is depicted by the blue and orange regions, while $R^c$ is denoted by gray. $r = \dist(Y,R^c)$.}
\end{figure}

We now use the Lieb-Robinson bound to construct piecewise time-independent approximations to the terms in the Hamiltonian. Define the following piecewise time-independent local term $h_X^\text{pw}(t)$ for $X$ as
\begin{equation}
    h_X^\text{pw}(t) = h_X\left(\frac{\lfloor m t\rfloor}{m}\right)
\end{equation} 
where $m$ is a parameter discretizing the time over the interval $0 < t < 1$. Since $h_X(t)$ is smooth, there exists a constant $\epsilon_\text{pw}$ such that for all $X,t$, $\norm{h_X^\text{pw}(t) - h_X(t)}\leq \epsilon_\text{pw}$ where $\epsilon_\text{pw}$ depends on $m$ (ultimately, the depth of the circuit).

Let $O_Y$ be a local observable supported on a set $Y$ and consider a partition of $\Lambda = R \cup R^c$ ($R \cap R^c = \emptyset$) such that $Y \subset R$. This decomposition is depicted in Fig.~\ref{fig:local_observable}, where $O_Y$ is supported in the orange region, and $R$ is supported in the blue and orange regions. Define the following partially piecewise time-independent Hamiltonian $H_R(t)$:

\begin{equation}
    H_R(t) = \sum_{X \subset R} h_X(t) + \sum_{X : X \cap R^c \neq \emptyset}h_X^{\text{pw}}(t)
    \label{eq:partitioned_H}
\end{equation}

We now bound the error between evolution by $H(t)$ and $H_R(t)$. For conciseness, we write $U_H(t)$ as $U_H^{t}$ and suppress time-dependence of $H$:
\begin{align}
&\norm{U_H^{t\dagger} O_Y U_H^t - U_{H_R}^{t\dagger} O_Y U_{H_R}^t}    \\
&= \bignorm{\int_0^t dt' \partial_{t'} U_H^{t-t'\dagger} U_{H_R}^{t'\dagger}O_Y U_{H_R}^{t'} U_{H_R}^{t-t'}} \\
&\leq \int_0^t dt' \norm{U_H^{t-t'\dagger} [H - H_R, U_{H_R}^{t'\dagger} O_Y U_{H_R}^{t'}]U_H^{t-t'}} \\
&= \int_0^t dt' \sum_{X : X \cap R^c \neq \emptyset} \norm{[h^{\text{pw}}_X - h_X, U_{H_R}^{t'\dagger} O_Y U_{H_R}^{t'}]}
\end{align}
where triangle inequality and unitary invariance of the norm have been used. From the last line, we apply Lemma~\ref{lemma:lieb_robinson} and sum over all $X$ to obtain that for any $t$, there exists a finite value $r_0$ such that for $r \geq r_0$ and $R$ chosen such that $\dist(Y, R^c) = r$
\begin{equation}
    \norm{U_H^{t\dagger} O_Y U_H^t - U_{H_R}^{t\dagger} O_Y U_{H_R}^t} \leq O(\norm{O_Y}|X|\epsilon_\text{pw} e^{-\mu r})
    \label{eq:first_sampling_approx}
\end{equation}

Finally, we analyze the error due to approximating terms within $R$ by their piecewise-constant counterparts. To do this, we use the following lemma:

\begin{lemma}[Unitary Approximation]\label{lemma:unitary_approx} Let $A(t)$ and $B(t)$ be two time-dependent Hermitian operators such that 

\begin{equation}
\max_{0\leq t' \leq t} \norm{B(t') - A(t')} \leq \epsilon
\end{equation}
Define unitary $U_A(t)$ as in Eq.~\ref{eq:time_dep_unitary}. Then

\begin{eqnarray}
    &\norm{U_B(t) - U_A(t)} \leq \epsilon  t \\
    &\norm{U_B(t)O U_B^\dagger(t) - U_A(t) O U_A^\dagger(t)} \leq 2\epsilon t\norm{O} \label{eq:operator_approx}
\end{eqnarray}
\end{lemma}

\begin{proof}
Define $\Delta(t) = U_A^\dagger(t)(B(t) - A(t)) U_A(t)$ and $U_\Delta(t)$ as above. Then it can be checked that $U_\Delta(t) = U_A^\dagger(t)U_B(t)$. Now
\begin{equation}
    \begin{aligned}
    \norm{U_B(t) - U_A(t)} &= \norm{U_\Delta(t) - U_\Delta(0)}\leq \epsilon \cdot t
    \end{aligned}
\end{equation}
Eq.~\ref{eq:operator_approx} follows similarly.
\end{proof}

We now construct a completely piecewise Hamiltonian $H^\text{pw}(t)$:

\begin{equation}
    H^{\text{pw}}(t) = \sum_{X \subset R} h_X^{\text{pw}}(t) + \sum_{X : X \cap R^c \neq \emptyset}h_X^{\text{pw}}(t)
\end{equation}
Using Eq.~\ref{eq:operator_approx}, the error from this step is $O(\norm{O_Y}|R|\epsilon_{\text{pw}})$. By adjusting $m$, this error along with that from Eq.~\ref{eq:first_sampling_approx} is reduced.

Finally, each piece of  $U_\text{pw}(t)$ may be Trotterized as
\begin{equation}
    U_\text{pw}(t) = \exp\left(-i\sum_{\{\mathcal{S}\}}\sum_{X\in \mathcal{S}}\frac{h_X^{\text{pw}}}{n}\right)^{n}  \approx \left(\prod_{\{\mathcal{S}\}} e^{-iH_\mathcal{S}/n}\right)^n
\end{equation}
where $\{\mathcal{S}\}$ is a partition of the sets $X$ such that $X\cap W=\emptyset$ for all sets $X\neq W$ within a partition $\mathcal{S}$, and $H_\mathcal{S} = \sum_{X\in \mathcal{S}}h_X$. Each $H_\mathcal{S}$ is a local commuting Hamiltonian. The preceding analysis is then applied again to the error for each $U_\text{pw}$ to see that local errors again dominate. By increasing $m$ and $n$, the error may be lowered arbitrarily and independently of the system size, up to the error set by the choice of $H(t)$ in the original quasiadiabatic evolution. The total depth of the circuit  then is $mn\mtude{\{\mathcal{S}\}}$, independent of system size. The main theorem is thus proved. We conclude with some comments. For brevity, we have concentrated on strictly local Hamiltonians. However, this analysis may be extended to Hamiltonians satisfying more general locality criteria. It may also be important to understand the errors arising for other operators, such as correlation functions and string operators.


\bibliography{main}
\end{document}